\documentclass[a4paper]{amsart}
\pdfoutput=1
\usepackage[utf8]{inputenc}
\usepackage[T1]{fontenc}
\usepackage{lmodern}
\usepackage{amssymb}
\usepackage{nicefrac,mathtools,enumitem,amsmath,amsxtra}
\usepackage{microtype}

\usepackage{tikz}
\usetikzlibrary{intersections}
\usetikzlibrary{matrix,arrows}
\tikzset{cd/.style=matrix of math nodes,row sep=2em,column sep=2em, text height=1.5ex, text depth=0.5ex}
\tikzset{cdar/.style=->,auto}
\usepackage{tikz-cd}

\tikzset{dar/.style={double,double equal sign distance,-implies}}
\tikzset{mid/.style={anchor=mid}} 

\tikzset{triar/.style={anchor=mid,->}}
\tikzset{tridar/.style={anchor=mid,double,double equal sign distance,-implies}}
\tikzset{narrowfill/.style={inner sep=0pt, fill=white}}

\setlist[enumerate,1]{label=\textup{(\arabic*)}}
\setlist[enumerate,2]{label=\textup{(\alph*)}}

\usepackage[pdftitle={Coarse geometry and topological phases},
pdfauthor={Ralf Meyer},
pdfsubject={Mathematics}
]{hyperref}
\usepackage[lite]{amsrefs}

\renewcommand{\PrintDOI}[1]{\href{http://dx.doi.org/\detokenize{#1}}{doi: \detokenize{#1}}}

\BibSpec{book}{%
  +{}  {\PrintPrimary}                {transition}
  +{,} { \textit}                     {title}
  +{.} { }                            {part}
  +{:} { \textit}                     {subtitle}
  +{,} { \PrintEdition}               {edition}
  +{}  { \PrintEditorsB}              {editor}
  +{,} { \PrintTranslatorsC}          {translator}
  +{,} { \PrintContributions}         {contribution}
  +{,} { }                            {series}
  +{,} { \voltext}                    {volume}
  +{,} { }                            {publisher}
  +{,} { }                            {organization}
  +{,} { }                            {address}
  +{,} { \PrintDateB}                 {date}
  +{,} { }                            {status}
  +{}  { \parenthesize}               {language}
  +{}  { \PrintTranslation}           {translation}
  +{;} { \PrintReprint}               {reprint}
  +{.} { }                            {note}
  +{.} {}                             {transition}
  +{} { \PrintDOI}                   {doi}
  +{} { available at \url}            {eprint}
  +{}  {\SentenceSpace \PrintReviews} {review}
}

\theoremstyle{plain}
\newtheorem{theorem}[subsection]{Theorem}
\newtheorem{lemma}[subsection]{Lemma}
\newtheorem{proposition}[subsection]{Proposition}

\newtheorem{corollary}[subsection]{Corollary}

\theoremstyle{definition}
\newtheorem{definition}[subsection]{Definition}

\theoremstyle{remark}

\newtheorem{example}[subsection]{Example}

\newcommand*{\defeq}{\mathrel{\vcentcolon=}}
\newcommand*{\congto}{\xrightarrow\sim}

\DeclarePairedDelimiter{\abs}{\lvert}{\rvert}
\DeclarePairedDelimiter{\norm}{\lVert}{\rVert}
\DeclarePairedDelimiterX{\braket}[2]{\langle}{\rangle}{#1\,\delimsize\vert\,\mathopen{}#2}
\DeclarePairedDelimiterX{\braketop}[3]{\langle}{\rangle}{#1\,\delimsize\vert #2\delimsize\vert\,\mathopen{}#3}
\DeclarePairedDelimiterX{\BRAKET}[2]{\langle}{\rangle}{\!\delimsize\langle#1\,\delimsize\vert\,\mathopen{}#2\delimsize\rangle\!}
\DeclarePairedDelimiterX{\setgiven}[2]{\{}{\}}{#1\,{:}\,\mathopen{}#2}

\newcommand{\idealin}{\mathrel{\triangleleft}} 
\newcommand*{\into}{\rightarrowtail}
\newcommand*{\injto}{\hookrightarrow}
\newcommand*{\prto}{\twoheadrightarrow}

\newcommand{\C}{\mathbb{C}}
\newcommand{\N}{\mathbb{N}}
\newcommand{\Z}{\mathbb{Z}}
\newcommand{\R}{\mathbb{R}}
\newcommand{\T}{\mathbb{T}}

\newcommand*{\Hils}[1][H]{\mathcal #1}

\newcommand*{\Cont}{\mathrm C}
\newcommand*{\Contc}{\mathrm{C}_\mathrm{c}}
\newcommand*{\Contb}{\mathrm{C}_\mathrm{b}}
\newcommand*{\Roe}{\mathrm{Roe}}

\newcommand*{\Star}{$^*$\nobreakdash-}
\newcommand*{\nb}{\nobreakdash}
\newcommand*{\Cst}{\mathrm C^*}
\newcommand*{\Cred}{\mathrm C^*_\mathrm r}

\newcommand*{\diff}{\mathrm{d}}
\newcommand*{\ima}{\mathrm i}
\newcommand*{\Euler}{\mathrm{e}}

\newcommand{\Comp}{\mathbb K}
\newcommand{\Bound}{\mathbb B}
\newcommand{\Mat}{\mathbb M}

\makeatletter
\newsavebox\myboxA
\newsavebox\myboxB
\newlength\mylenA

\newcommand*\xoverline[2][0.75]{%
    \sbox{\myboxA}{$\m@th#2$}%
    \setbox\myboxB\null
    \ht\myboxB=\ht\myboxA%
    \dp\myboxB=\dp\myboxA%
    \wd\myboxB=#1\wd\myboxA
    \sbox\myboxB{$\m@th\overline{\copy\myboxB}$}
    \setlength\mylenA{\the\wd\myboxA}
    \addtolength\mylenA{-\the\wd\myboxB}%
    \ifdim\wd\myboxB<\wd\myboxA%
       \rlap{\hskip 0.5\mylenA\usebox\myboxB}{\usebox\myboxA}%
    \else
        \hskip -0.5\mylenA\rlap{\usebox\myboxA}{\hskip 0.5\mylenA\usebox\myboxB}%
    \fi}
\makeatother

\newcommand*{\conj}[1]{\overline{#1}}

\newcommand*{\Id}{\mathrm{id}}
\newcommand*{\K}{\mathrm{K}}
\newcommand*{\KR}{\mathrm{KR}}
\newcommand*{\KK}{\mathrm{KK}}

\DeclareMathOperator{\supp}{supp}

\DeclareMathOperator{\Aut}{Aut}
\DeclareMathOperator{\Ad}{Ad}

\DeclareMathOperator{\Cliff}{Cl}

\begin{document}
\title{Coarse geometry and topological phases}
\author{Eske Ellen Ewert}
\email{Eske.Ewert@mathematik.uni-goettingen.de}

\author{Ralf Meyer}
\email{rmeyer2@uni-goettingen.de}

\address{Mathematisches Institut\\
  Universit\"at G\"ottingen\\Bunsenstra\ss e 3--5\\
  37073 G\"ottingen\\Germany}



\begin{abstract}
  We propose the Roe \(\Cst\)\nb-algebra
  from coarse geometry as a model for topological phases of
  disordered materials.  We explain the robustness of this
  \(\Cst\)\nb-algebra
  and formulate the bulk--edge correspondence in this framework.  We
  describe the map from the K\nb-theory of the group
  \(\Cst\)\nb-algebra
  of~\(\Z^d\)
  to the K\nb-theory of the Roe \(\Cst\)\nb-algebra,
  both for real and complex K\nb-theory.
\end{abstract}
\maketitle

\section{Introduction}
\label{sec:intro}

Topological insulators are materials that are insulating in the bulk
but allow a current to flow on the boundary.  These boundary currents
are protected by topological invariants and thus, in ideal cases, flow
without dissipation.  The mathematical description of a topological
insulator uses a \(\Cst\)\nb-algebra~\(\mathcal{A}\)
that contains the resolvent of the Hamiltonian~\(H\)
of the system; this amounts to \(H\in\mathcal{A}\)
if~\(H\)
is bounded.  To describe an insulator, the spectrum of~\(H\)
should have a gap at the Fermi energy~\(E\).
Depending on further symmetries of the system such as a time
reversal, particle--hole or chiral symmetry, the topological phase
of the material may be classified by a class in the \(\K\)\nb-theory
of~\(\mathcal{A}\)
associated to the spectral projection of~\(H\)
at the Fermi energy (see, for instance,
\cites{Kellendonk:Cstar_phases, Prodan-Schulz-Baldes:Bulk_boundary}).  
So the observable algebra~\(\mathcal{A}\)
or rather its K\nb-theory predicts the possible topological
phases of a material.

At first, a material is often modelled without disorder and in a tight
binding approximation.  This gives a translation-invariant Hamiltonian
acting on \(\ell^2(\Z^d,\C^N)\)
(see, for instance, \cites{Bernevig-Hughes-Zhang:Quantum,
  Fu-Kane-Mele:Insulators,
  Liu-Qi-Zhang-Dai-Fang-Zhang:Model_Hamiltonian}).  Bloch--Floquet
theory describes the Fermi projection through a vector bundle over the
\(d\)\nb-torus,
with extra structure that reflects the symmetries of the system (see,
for instance, \cites{De_Nittis-Gomi:Real_Bloch,
  De_Nittis-Gomi:Quaternionic_Bloch, Kennedy-Zirnbauer:Bott_gapped}).
The K\nb-theory of the \(d\)\nb-torus
is easily computed.  Once \(d\ge2\),
many of the topological phases that are predicted this way are
obtained by stacking a lower-dimensional topological insulator in some
direction.  Such topological phases are called ``weak'' by
Fu--Kane--Mele~\cite{Fu-Kane-Mele:Insulators}.  They claim that
weak topological phases are not robust under disorder.

Other authors have claimed instead that weak topological insulators
are also quite robust, see~\cite{Ringel-Kraus-Stern:Strong_side}.
Their proof of robustness, however, is no longer topological.
Roughly speaking, the idea is that, although disorder may destroy
the topological phase, it must be rather special to do this.
\emph{Random} disorder will rarely be so special.  So in a finite
volume approximation, the topological phase will remain intact in
most places, and the small area where the randomness destroys it
will become negligible in the limit of infinite volume.  Such an
argument may also work for the Hamiltonian of an insulator that is
homotopic to a trivial one.  Our study is purely topological in
nature and thus cannot see such phenomena.

We are going to explain the difference between strong and weak
topological phases and the robustness of the former through a
difference in the underlying observable algebras.  Namely, we shall
model a material with disorder by the Roe \(\Cst\)\nb-algebra
of \(\R^d\)
or~\(\Z^d\),
which is a central object of coarse geometry.
Roe~\cites{Roe:Index_open_I, Roe:Index_open_II} introduced them
to get index theorems for elliptic operators on non-compact
Riemannian manifolds.

Before choosing our observable algebra, we should ask: What is causing
topological phases?  At first sight, the answer seems to be the
translation invariance of the Hamiltonian.  Translation-invariance
alone is not enough, however.  And it is destroyed by disorder.  The
subalgebra of translation-invariant operators on the Hilbert space
\(\ell^2(\Z^d,\C^N)\)
is the algebra of \(N\times N\)-matrices
over the group von Neumann algebra of~\(\Z^d\),
which is isomorphic to \(L^\infty(\T^d, \Mat_N)\).
If topological phases were caused by translation invariance alone,
they should be governed by the K\nb-theory of \(L^\infty(\T^d, \Mat_N)\).
This is clearly not the case.  Instead, we need the group
\(\Cst\)\nb-algebra,
which is isomorphic to \(\Cont(\T^d,\Mat_N)\).
The reason why the spectral projections of the Hamiltonian belong to
the group \(\Cst\)\nb-algebra
instead of the group von Neumann algebra is that the matrix
coefficients of the Hamiltonian for \((x,y)\in\Z^d\)
are supported in the region \(\norm{x-y}\le R\)
for some \(R>0\);
let us call such operators \emph{controlled}.  The controlled
operators do not form a \(\Cst\)\nb-algebra,
and it makes no difference for \(\K\)\nb-theory purposes to allow
the Hamiltonian to be a limit of controlled operators in the norm
topology.  We shall see below that this is equivalent to continuity
with respect to the action of~\(\R^d\)
on \(\Bound(\ell^2(\Z^d,\C^N))\)
generated by the position observables.  This action restricts to
the translation action of~\(\R^d\)
on the group von Neumann algebra \(L^\infty(\T^d, \Mat_N)\),
so that its continuous elements are the functions in
\(\Cont(\T^d,\Mat_N)\).
Hence \(\Cont(\T^d,\Mat_N) \subseteq \Bound(\ell^2(\Z^d,\C^N))\)
consists of those operators that are both translation-invariant and
norm limits of controlled operators. 

Since disorder destroys translation invariance, we should drop this
assumption to model systems with disorder.  The \(\Cst\)\nb-algebra
of all operators on \(\ell^2(\Z^d,\C^N)\)
that are norm limits of controlled operators is the algebra
of \(N\times N\)-matrices
over the \emph{uniform Roe \(\Cst\)\nb-algebra}
of~\(\Z^d\).
To get the Roe \(\Cst\)\nb-algebra,
we must work in \(\ell^2(\Z^d,\Hils)\)
for a separable Hilbert space~\(\Hils\)
and add a local compactness property, namely, that the operators
\(\braketop{x}{(H-\lambda)^{-1}}{y} \in\Bound(\Hils)\)
are compact for all \(x,y\in\Z^d\),
\(\lambda\in \C\setminus\R\).
This property is automatic for operators on \(\ell^2(\Z^d,\C^N)\).
Working on \(\ell^2(\Z^d,\Hils)\)
and assuming local compactness means that we include infinitely many
bands in our model and require only finitely many states with finite
energy in each finite volume.
Kubota~\cite{Kubota:Controlled_bulk-edge} has already used the
uniform Roe \(\Cst\)\nb-algebra
and the Roe \(\Cst\)\nb-algebra
in the context of topological insulators.  He prefers the uniform
Roe \(\Cst\)\nb-algebra.  We explain why we consider this a
mistake.

Working in the Hilbert space \(\ell^2(\Z^d,\C^N)\)
already involves an approximation.  We ought to work in a continuum
model, that is, in the Hilbert space \(L^2(\R^d,\C^k)\),
where~\(k\)
is the number of internal degrees of freedom.
This Hilbert space is isomorphic to
\[
L^2(\R^d,\C^k)
\cong L^2(\Z^d\times (0,1]^d)\otimes \C^k
\cong \ell^2(\Z^d,\Hils\otimes \C^k)
\]
when we cover~\(\R^d\)
by the disjoint translates of the fundamental domain~\((0,1]^d\).
This identification preserves both controlled and locally compact
operators.  Thus the Roe \(\Cst\)\nb-algebras
of \(\Z^d\)
and~\(\R^d\)
are isomorphic.  For the Roe \(\Cst\)\nb-algebra
of~\(\R^d\),
it makes no difference to replace \(L^2(\R^d)\)
by \(L^2(\R^d,\C^k)\)
or \(L^2(\R^d,\Hils)\):
all these Hilbert spaces give isomorphic \(\Cst\)\nb-algebras
of locally compact, approximately controlled operators.  So
there is only one Roe \(\Cst\)\nb-algebra
for~\(\R^d\),
and it is isomorphic to the non-uniform Roe \(\Cst\)\nb-algebra
of~\(\Z^d\).
We view the appearance of the uniform Roe \(\Cst\)\nb-algebra
for~\(\Z^d\)
as an artefact of simplifying assumptions in tight binding models.

We describe some interesting elements of the Roe \(\Cst\)\nb-algebra
of~\(\R^d\)
in Example~\ref{exa:Roe_Rn}.  In particular, it contains all
\(\Cont_0\)\nb-functions of the impulse operator~\(P\)
on \(L^2(\R^d)\) or, equivalently,
\begin{equation}
  \label{eq:f_of_P}
  \int_{\R^d} f(x) \exp(\ima x P) \,\diff x  
\end{equation}
for \(f\in \Cst(\R^d)\);
this operator is controlled if and only if~\(f\)
has compact support.  If \(V\in L^\infty(\R^d)\),
then the operator of multiplication by~\(V\) on~\(L^2(\R^d)\)
is controlled, but not locally compact.  Its product with an
operator as in~\eqref{eq:f_of_P} belongs to the Roe
\(\Cst\)\nb-algebra.

The real and complex K\nb-theory of the Roe \(\Cst\)\nb-algebra
of~\(\Z^d\)
is well known: up to a dimension shift of~\(d\),
it is the \(\K\)\nb-theory
of \(\R\)
or~\(\C\),
respectively.  In particular, the Roe \(\Cst\)\nb-algebra
as an observable algebra is small enough to predict some distinct
topological phases.  These coincide with Kitaev's periodic
table~\cite{Kitaev:Periodic_table}.  This corroborates the choice of
the Roe \(\Cst\)\nb-algebra
as the observable algebra for disordered materials.

When we disregard disorder, the Roe \(\Cst\)\nb-algebra
may be replaced by its translation-invariant subalgebra, which is
isomorphic to
\[
\Cst(\Z^d) \otimes \Comp(\Hils) \cong \Cont(\T^d,\Comp(\Hils)),
\]
where \(\Comp(\Hils)\)
denotes the \(\Cst\)\nb-algebra
of compact operators on an infinite-dimensional separable Hilbert
space~\(\Hils\).
In the real case, the \(d\)\nb-torus
must be given the real involution by the restriction of complex
conjugation on \(\C^d\supseteq\T^d\).
The real or complex \(\K\)\nb-theory
groups of the ``real'' \(d\)\nb-torus describe both weak and strong
topological phases in the presence of different types of symmetries.
We show that the map
\begin{equation}
  \label{eq:comparison_group_Roe}
  \K_*(\Cst(\Z^d)_\mathbb{F}) \to \K_*(\Cst_\Roe(\Z^d)_\mathbb{F})
\end{equation}
for \(\mathbb{F} = \R\)
or \(\mathbb{F} = \C\)
is split surjective and that its kernel is the subgroup generated by
the images of \(\K_*(\Cst(\Z^{d-1})_\mathbb{F})\)
for all coordinate embeddings \(\Z^{d-1} \to \Z^d\).
That is, the kernel of the map in~\eqref{eq:comparison_group_Roe}
consists exactly of the \(\K\)\nb-theory
classes of weak topological insulators as defined by
Fu--Kane--Mele~\cite{Fu-Kane-Mele:Insulators}.  The strong topological
insulators are those that remain topologically protected even if the
observable algebra is enlarged to the Roe \(\Cst\)\nb-algebra,
allowing rather general disorder.

Following Bellissard \cites{Bellissard:K-theory_solid,
  Bellissard-Elst-Schulz-Baldes:Noncommutative_Hall}, disorder is
usually modelled by crossed product \(\Cst\)\nb-algebras
\(\mathcal{A} = \Cont(\Omega)\rtimes\Z^d\),
where~\(\Omega\)
is the space of disorder configurations.  It is more precise to say,
however, that the space~\(\Omega\)
describes \emph{restricted} disorder.  Uncountably many different
choices are possible.  Such models are only reasonable when the
physically relevant objects do not depend on the choice.  But the
\(\K\)\nb-theory
of the crossed product depends on the topology of the
space~\(\Omega\).
To make it independent of~\(\Omega\),
the space~\(\Omega\)
is assumed to be contractible
in~\cite{Prodan-Schulz-Baldes:Bulk_boundary}.
This fits well with standard choices of~\(\Omega\)
such as a product space \(\prod_{n\in\Z^d} [-1,1]\)
to model a random potential.  The resulting \(\K\)\nb-theory
then becomes the same as in the system without disorder.  So another
argument must be used to explain the difference between weak and
strong topological phases, compare
\cite{Prodan-Schulz-Baldes:Bulk_boundary}*{Remark 5.3.5}.
If one allows non-metrisable~\(\Omega\), then
there is a maximal choice for~\(\Omega\), namely,
the Stone--\v{C}ech compactification of~\(\Z^d\).
The resulting crossed product \(\ell^\infty(\Z^d)\rtimes\Z^d\)
is isomorphic to the uniform Roe \(\Cst\)\nb-algebra
of~\(\Z^d\),
see~\cite{Kubota:Controlled_bulk-edge}.
Nevertheless, even this maximal choice of~\(\Omega\) still contains
a hidden restriction on disorder: the number of bands for a tight
binding model is fixed, and so the disorder is also limited to a
fixed finite number of bands.  The Roe \(\Cst\)\nb-algebra
of~\(\Z^d\)
also removes this hidden restriction on the allowed disorder.
It is also a crossed product, namely,
\[
\Cst_\Roe(\Z^d) \cong \ell^\infty(\Z^d,\Comp(\Hils)) \rtimes \Z^d.
\]
The \(\Cst\)\nb-algebra \(\ell^\infty(\Z^d,\Comp(\Hils))\)
is not isomorphic to \(\ell^\infty(\Z^d) \otimes \Comp(\Hils)\):
it even has different \(\K\)\nb-theory.

Since the Roe \(\Cst\)\nb-algebra
has not been used much in the context of topological insulators, we
recall its main properties in Section~\ref{sec:Roe_Cstar}.  We
highlight its robustness or even ``universality.''  Roughly
speaking, there is only one Roe
\(\Cst\)\nb-algebra
in each dimension, which describes all kinds of disordered materials
in that dimension without symmetries.  The various symmetries
(time-reversal, particle-hole, chiral) may be added by tensoring the
real or complex Roe \(\Cst\)\nb-algebra
with Clifford algebras, which replaces \(\K_0\)
by~\(\K_i\)
for some \(i\in\Z\).
We shall not say much about this here.  The Roe \(\Cst\)\nb-algebra
of a coarse space is a coarse invariant.  In particular, all coarsely
dense subsets in~\(\R^d\)
give isomorphic Roe \(\Cst\)\nb-algebras.
Furthermore, the Roe \(\Cst\)\nb-algebras
of~\(\Z^d\)
and other coarsely dense subsets of~\(\R^d\)
are isomorphic to that of~\(\R^d\).
Thus it makes no difference whether we work in a continuum or lattice
model.  We also consider the twists of the Roe \(\Cst\)\nb-algebra
defined by magnetic fields.  The resulting twisted Roe
\(\Cst\)\nb-algebras
are also isomorphic to the untwisted one.  This robustness of the Roe
\(\Cst\)\nb-algebra
means that the same \emph{strong} topological phases occur for all
materials of a given dimension and symmetry type, even for
quasi-crystals and aperiodic materials.

We compute the \(\K\)\nb-theory
of the Roe \(\Cst\)\nb-algebra
in Section~\ref{sec:coarse_MV} using the coarse Mayer--Vietoris
principle introduced in~\cite{Higson-Roe-Yu:Coarse_Mayer-Vietoris}.
We prove the Mayer--Vietoris exact sequence in the real and complex
case by reducing it to the \(\K\)\nb-theory
exact sequence for \(\Cst\)\nb-algebra
extensions.  The computation of \(\K_*(\Cst_\Roe(X)_\mathbb{F})\)
for \(X=\Z^d\)
is based on the vanishing of this invariant for half-spaces
\(\N\times\Z^{d-1}\).
This implies
\(\K_{*+d}(\Cst_\Roe(X)_\mathbb{F}) \cong \K_*(\mathbb{F})\).
It shows also that the inclusion
\(\Cst_\Roe(\Z^{d-1})_\mathbb{F} \to \Cst_\Roe(\Z^d)_\mathbb{F}\)
induces the zero map on \(\K\)\nb-theory.
Hence the map~\eqref{eq:comparison_group_Roe} kills all
elements of \(\K_*(\Cst(\Z^d)_\mathbb{F})\)
that come from inclusions \(\Z^{d-1} \injto \Z^d\).

In Section~\ref{sec:weak_phases} we compute the real and complex
\(\K\)\nb-theory for the group \(\Cst\)\nb-algebra
of~\(\Z^d\)
and the map in~\eqref{eq:comparison_group_Roe}.  More precisely, we
describe the composite of the map in~\eqref{eq:comparison_group_Roe}
with the isomorphism
\(\K_{*+d}(\Cst_\Roe(X)_\mathbb{F}) \cong \K_*(\mathbb{F})\):
this is the pairing with the fundamental class of the ``real''
\(d\)\nb-torus~\(\T^d\).
Except for an adaptation to ``real'' manifolds, this fundamental
class is introduced in~\cite{Kasparov:Novikov}.  We show that the
fundamental class extends to a \(\K\)\nb-homology
class on the Roe \(\Cst\)\nb-algebra
and that the pairing with this \(\K\)\nb-homology
class is an isomorphism
\(\K_{*+d}(\Cst_\Roe(X)_\mathbb{F}) \cong \K_*(\mathbb{F})\).

\section{Roe \texorpdfstring{$\Cst$}{C*}-algebras}
\label{sec:Roe_Cstar}

In this section, we define the real and complex Roe
\(\Cst\)\nb-algebras
of a proper metric space and prove that they are invariant under
passing to a coarsely dense subspace and, more generally, under coarse
equivalence.  We prove that the twists used to encode magnetic fields
do not change them.  And we describe elements of the Roe
\(\Cst\)\nb-algebra
of a subset of~\(\R^d\)
as those locally compact operators that are continuous for the
representation of~\(\R^d\)
generated by the position operators.  We describe the subalgebras of
smooth, real-analytic and holomorphic elements of the Roe
\(\Cst\)\nb-algebra
for this action of~\(\R^d\).
We show that Roe \(\Cst\)\nb-algebras
have approximate units of projections, which simplifies the definition
of their \(\K\)\nb-theory.

Let \((X,d)\)
be a locally compact, second countable, metric space.  We assume the
metric~\(d\)
to be \emph{proper}, that is, bounded subsets of~\(X\)
are compact.  We shall be mainly interested in \(\R^d\)
or a discrete subset of~\(\R^d\)
with the restriction of the Euclidean metric.  (All our results on
general proper metric spaces extend easily to the more general coarse
spaces introduced in~\cite{Roe:Lectures}.)  Let~\(\Hils\)
be a real or complex separable Hilbert space and let
\(\varrho\colon \Cont_0(X)\to\Bound(\Hils)\)
be a nondegenerate representation.  We are going to define the Roe
\(\Cst\)\nb-algebra
of~\(X\)
with respect to~\(\varrho\),
see also \cite{Higson-Roe:Analytic_K}*{Section 6.3}.  Depending on
whether~\(\Hils\)
is a real or complex Hilbert space, this gives a real or complex
version of the Roe \(\Cst\)\nb-algebra.
Both cases are completely analogous.

Let \(T\in\Bound(\Hils)\).
We call~\(T\)
\emph{locally compact} (on~\(X\))
if the operators \(\varrho(f)T\)
and~\(T\varrho(f)\)
are compact for all \(f\in\Cont_0(X)\).
The \emph{support} of~\(T\)
is a subset \(\supp T\subseteq X\times X\).
Its complement consists of all \((x,y)\in X\times X\)
for which there are neighbourhoods \(U_x\),
\(U_y\)
in~\(X\)
such that \(\varrho(f) T \varrho(g)=0\)
for all \(f\in\Cont_0(U_x)\),
\(g\in\Cont_0(U_y)\).
The operator~\(T\)
is \emph{controlled} (or has \emph{finite propagation}) if there is
\(R>0\)
such that \(d(x,y)\le R\)
for all \((x,y)\in \supp T\).
We sometimes write ``\(R\)\nb-controlled'' to highlight the control
parameter~\(R\).
The locally compact, controlled operators on~\(\Hils\)
form a \Star{}algebra.  Its closure in \(\Bound(\Hils)\)
is the \emph{Roe \(\Cst\)\nb-algebra} \(\Cst_\Roe(X,\varrho)\).

The representation~\(\varrho\)
is called \emph{ample} if the operator \(\varrho(f)\)
for \(f\in\Cont_0(X)\)
is only compact for \(f=0\).

\begin{theorem}
  \label{the:Roe_ample_unique}
  Let \(\varrho_i\colon \Cont_0(X)\to\Bound(\Hils_i)\)
  for \(i=1,2\)
  be ample representations, where \(\Hils_1\)
  and~\(\Hils_2\)
  are both complex or both real.  Then
  \(\Cst_\Roe(X,\varrho_1) \cong \Cst_\Roe(X,\varrho_2)\).
  Even more, there is a unitary operator
  \(U\colon \Hils_1 \congto \Hils_2\)
  with
  \[
  U \Cst_\Roe(X,\varrho_1) U^* = \Cst_\Roe(X,\varrho_2).
  \]
\end{theorem}

Many references only assert the weaker statement that the Roe
\(\Cst\)\nb-algebras
for all ample representations have canonically isomorphic
\(\K\)\nb-theory,
compare \cite{Higson-Roe:Analytic_K}*{Corollary 6.3.13}.  The
statement above is
\cite{Higson-Roe-Yu:Coarse_Mayer-Vietoris}*{Lemma~2}, and our proof is
the same.

\begin{proof}
  If~\(X\)
  is compact, then \(\Cst_\Roe(X,\varrho_i) = \Comp(\Hils_i)\).
  Since~\(\Hils_i\)
  for \(i=1,2\)
  are assumed to be separable, there is a unitary
  \(U\colon \Hils_1 \congto \Hils_2\),
  and it will do the job.  So we may assume~\(X\)
  to be non-compact.  Fix \(R>0\).
  The open balls \(B(x,R)\)
  for \(x\in X\)
  cover~\(X\).
  Since~\(X\)
  is second countable, there is a subordinate countable, locally
  finite, open covering \(X= \bigcup_{n\in\N} U_n\),
  where each~\(U_n\)
  is non-empty and has diameter at most~\(R\).
  Then there is a countable covering of~\(X\)
  by disjoint Borel sets, \(X=\bigsqcup_{n\in\N} B'_n\),
  where each~\(B'_n\)
  has diameter at most~\(R\),
  and such that any relatively compact subset is already covered by
  finitely many of the~\(B'_n\):
  simply take \(B'_n \defeq U_n \setminus \bigcup_{j< n} U_j\).
  Next we modify the subsets~\(B'_n\)
  so that they all have non-empty interior.  Let \(M\subseteq\N\)
  be the set of all \(n\in\N\)
  for which~\(B'_n\)
  has non-empty interior.  Let \(m\in M\).
  Let~\(K_m\)
  be the set of all \(k\in\N\)
  for which~\(B'_k\)
  has empty interior and \(U_m\cap U_k\neq \emptyset\).
  The set~\(K_m\)
  is finite because~\(U_m\)
  is bounded and hence relatively compact.  Let
  \(K_m^\circ \defeq K_m \setminus \bigcup_{i\in M, i<m} K_i\).
  Define
  \[
  B_m \defeq B'_m \sqcup \bigsqcup_{k\in K_m^\circ} B'_k.
  \]
  This is a Borel set.  It has non-empty interior because~\(B'_m\)
  has non-empty interior.  Its diameter is at most~\(3R\)
  because all the~\(U_k\)
  for \(k\in K_m\)
  intersect~\(U_m\).
  The definition of~\(K_m^\circ\)
  ensures that the subsets~\(B_m\)
  for \(m\in M\)
  are disjoint.  We claim that \(\bigsqcup_{m\in M} B_m = X\).
  This is equivalent to
  \(\bigcup_{m\in M} K_m^\circ = \N\setminus M\)
  because \(B'_m\subseteq B_m\)
  for all \(m\in M\).
  This is further equivalent to
  \(\bigcup_{m\in M} K_m = \N\setminus M\).
  Let \(k_0\in\N \setminus M\),
  that is, \(B'_{k_0}\)
  has empty interior.  Then~\(B'_{k_0}\)
  does not contain~\(U_{k_0}\).
  So there is some \(k_1<k_0\)
  with \(U_{k_0} \cap U_{k_1} \neq \emptyset\).
  If \(k_1\in \N\setminus M\),
  then~\(B'_{k_1}\)
  does not contain \(U_{k_0} \cap U_{k_1}\).
  So there is \(k_2<k_1\)
  with \(U_{k_0} \cap U_{k_1} \cap U_{k_2} \neq \emptyset\).
  We continue like this and build a decreasing chain
  \(k_0>k_1>\dotsc>k_\ell\)
  such that \(k_0,\dotsc,k_{\ell-1}\in \N\setminus M\) and
  \[
  U_{k_0} \cap U_{k_1} \cap \dotsb \cap U_{k_\ell}
  \neq \emptyset.
  \]
  We eventually reach \(k_\ell\in M\)
  because \(B'_1= U_1\)
  is open and so \(1\in M\).
  We have \(k_0\in K_{k_\ell}\).
  So \(\bigcup_{m\in M} K_m = \N\setminus M\) as asserted.

  We have built a covering of~\(X\)
  by disjoint Borel sets \(X = \bigsqcup_{m\in M} B_m\)
  of diameter at most~\(3 R\),
  with non-empty interiors, and such that any relatively compact
  subset is already covered by finitely many of the~\(B_m\).
  The set~\(M\)
  is at most countable, and it cannot be finite because then~\(X\)
  would be bounded and hence compact.

  Using the Borel functional calculus for the
  representation~\(\varrho_i\),
  we may decompose the Hilbert space~\(\Hils_i\)
  as an orthogonal direct sum,
  \(\Hils_i = \bigoplus_{m\in M} \Hils_{i,m}\),
  where~\(\Hils_{i,m}\)
  is the image of the projection \(\varrho_i(1_{B_m})\).
  Since each~\(B_m\)
  has non-empty interior and our representations are ample, there is a
  non-compact operator on each~\(\Hils_{i,m}\).
  So no~\(\Hils_{i,m}\)
  has finite dimension.  Hence there is a unitary
  \(U_m\colon \Hils_{1,m}\congto\Hils_{2,m}\)
  for each \(m\in M\).
  We combine these into a unitary operator
  \(U= \bigoplus_{m\in M} U_m\colon \Hils_1 \congto \Hils_2\).

  Let \(T\in\Bound(\Hils_1)\).
  We claim that~\(U T U^*\)
  is locally compact or controlled if and only if~\(T\)
  is.  This implies
  \(U \Cst_\Roe(X,\varrho_1) U^* = \Cst_\Roe(X,\varrho_2)\)
  as asserted.  First, we claim that~\(T\)
  is locally compact if and only if \(T \varrho_1(1_{B_m})\)
  and \(\varrho_1(1_{B_m}) T\)
  are compact for all \(m\in M\).
  In one direction, this uses that there is \(g\in \Cont_0(X)\)
  with \(1_{B_m} \le g\)
  because~\(B_m\)
  has finite diameter.  In the other direction, it uses that any
  relatively compact subset of~\(X\)
  is already covered by finitely many~\(B_m\).
  Since \(U(\Hils_{1,m}) = \Hils_{2,m}\),
  the criterion above shows that~\(T\)
  is locally compact if and only if \(U T U^*\)
  is so.  Since the diameter of~\(B_m\)
  is at most~\(3 R\),
  the operator~\(U_m\),
  viewed as a partial isometry on \(\Hils_1 \oplus \Hils_2\),
  is \(3 R\)\nb-controlled.
  Thus~\(U\)
  is also \(3 R\)\nb-controlled.
  So \(U T U^*\) is controlled if and only if~\(T\) is.
\end{proof}

\begin{corollary}
  \label{cor:Roe_matrix-stable}
  Let~\(\varrho\)
  be an ample representation and let \(m\in\N_{\ge2}\).  Then
  \[
  \Cst_\Roe(X,\varrho) \cong \Mat_m(\Cst_\Roe(X,\varrho)).
  \]
\end{corollary}

\begin{proof}
  The direct sum representation \(m\cdot \varrho\)
  is still ample.  So
  \(\Cst_\Roe(X,m\cdot \varrho) \cong \Cst_\Roe(X,\varrho)\).
  An operator on~\(\Hils^m\)
  is locally compact or controlled if and only if its block
  matrix entries in \(\Bound(\Hils)\)
  are so.  Thus
  \(\Cst_\Roe(X,m\cdot \varrho) = \Mat_m(\Cst_\Roe(X,\varrho))\).
\end{proof}

The stabilisation \(\Cst_\Roe(X,\varrho) \otimes \Comp(\ell^2\N)\),
however, is usually not isomorphic to \(\Cst_\Roe(X,\varrho)\).

\begin{example}
  \label{exa:Roe_discrete}
  Let~\(X\)
  be discrete, for instance, \(X=\Z^d\).
  The representation~\(\varrho\)
  of \(\Cont_0(X)\)
  on~\(\ell^2(X)\)
  by multiplication operators is not ample.  It defines the
  \emph{uniform Roe \(\Cst\)\nb-algebra}
  of~\(X\).
  To get the Roe \(\Cst\)\nb-algebra,
  we may take the representation of \(\Cont_0(X)\)
  on \(\ell^2(X)\otimes \ell^2(\N)\).

  An operator~\(T\)
  on \(\ell^2(X)\otimes \ell^2(\N)\)
  is determined by its matrix coefficients
  \(T_{x,y} = \braketop{x}{T}{y} \in \Bound(\ell^2(\N))\)
  for \(x,y\in X\).
  It is locally compact if and only if all~\(T_{x,y}\)
  are compact.  Its support is the set of all \((x,y)\in X^2\)
  with \(T_{x,y}\neq0\).
  So it is controlled if and only if there is \(R>0\)
  so that \(T_{x,y}=0\)
  for \(d(x,y)>R\).
  The Roe \(\Cst\)\nb-algebra is the norm closure of these operators.

  If~\(X\)
  is a discrete group equipped with a translation-invariant metric,
  then \(\Cst_\Roe(X)\)
  is isomorphic to the reduced crossed product for the translation
  action of~\(X\)
  on~\(\ell^\infty(X, \Comp(\ell^2\N))\)
  (compare \cite{Roe:Lectures}*{Theorem~4.28} for the uniform Roe
  \(\Cst\)\nb-algebra).
\end{example}

\begin{example}
  \label{exa:Roe_Rn}
  Let \(X=\R^d\).
  The representation~\(\varrho\)
  of \(\Cont_0(\R^d)\)
  on \(L^2(\R^d,\diff x)\)
  (real or complex) by multiplication operators is ample.  Actually,
  all faithful representations of~\(\Cont_0(\R^d)\)
  are ample.  So they all give isomorphic Roe \(\Cst\)\nb-algebras
  by Theorem~\ref{the:Roe_ample_unique}.

  Let \(T\in\Contc(\R^d)\)
  (with real or complex values) act on~\(L^2(\R^d)\)
  by convolution.  Then \(T\cdot \varrho(f)\)
  and \(\varrho(f)\cdot T\)
  are compact because they have a compactly supported, continuous
  integral kernel.  And \(T\)
  is controlled by the supremum of~\(\norm{x}\)
  with \(T(x)\neq0\).
  So \(T\in\Cst_\Roe(\R^d)\).
  Hence \(\Cst(\R^d) \subseteq \Cst_\Roe(\R^d)\).
  In particular, the resolvent of the Laplace operator or another
  translation-invariant elliptic differential operator on~\(\R^d\)
  belongs to
  \(\Cst_\Roe(\R^d)\).
  Any multiplication operator is controlled.  Thus
  multiplication operators are multipliers of \(\Cst_\Roe(\R^d)\).
  And \(L^\infty(\R^d) \cdot \Cst(\R^d) \cdot L^\infty(\R^d)\)
  is contained in \(\Cst_\Roe(\R^d)\).
  (Since the translation action of~\(\R^d\)
  on \(L^\infty(\R^d)\)
  is not continuous, there is no crossed product for this action and
  it is unclear whether the closed linear spans of
  \(L^\infty(\R^d) \cdot \Cst(\R^d)\)
  and \(\Cst(\R^d) \cdot L^\infty(\R^d)\)
  are equal and form a \(\Cst\)\nb-algebra.)
\end{example}

\begin{proposition}
  \label{pro:resolvent_in_Roe}
  Let \(V\in L^\infty(\R^n)\)
  and let~\(\Delta\)
  be the Laplace operator on~\(\R^d\).
  Then the resolvent of \(V+\Delta\)
  belongs to the Roe \(\Cst\)\nb-algebra of~\(\R^d\).
\end{proposition}

We are indebted to Detlev Buchholz for pointing out the following
simple proof.

\begin{proof}
  View \(V=V(Q)\)
  as an operator on \(L^2(\R^d)\).
  Then \(\norm{(\ima c + \Delta)^{-1} V}^2<1\)
  for sufficiently large \(c\in\R_{>0}\).
  Hence the Neumann series \(\sum (-(\ima c + \Delta)^{-1} V)^n\)
  converges, and
  \[
  \sum_{n=0}^\infty (-(\ima c + \Delta)^{-1} V)^n
  \cdot (\ima c + \Delta)^{-1}
  = (1+(\ima c + \Delta)^{-1} V)^{-1} \cdot (\ima c + \Delta)^{-1}
  = (\ima c + \Delta+V)^{-1}.
  \]
  We have already seen that \((\ima c + \Delta)^{-1}\)
  and \((\ima c + \Delta)^{-1} V\)
  belong to the Roe \(\Cst\)\nb-algebra.
  Hence so does \((\ima c + \Delta+V)^{-1}\).
\end{proof}

If~\(\varrho\)
is ample, then we often leave out~\(\varrho\)
and briefly write \(\Cst_\Roe(X)_\R\)
or \(\Cst_\Roe(X)_\C\),
depending on whether~\(\varrho\)
acts on a real or complex Hilbert space.
Theorem~\ref{the:Roe_ample_unique} justifies this.

\begin{definition}
  \label{def:coarsely_dense}
  A closed subset \(Y\subseteq X\)
  is \emph{coarsely dense} if there is \(R>0\)
  such that for any \(x\in X\) there is \(y\in Y\) with \(d(x,y)\le R\).
\end{definition}

\begin{theorem}
  \label{the:coarsely_dense}
  Let \(Y\subseteq X\)
  be coarsely dense.  Then \(\Cst_\Roe(Y)_\R \cong \Cst_\Roe(X)_\R\)
  and \(\Cst_\Roe(Y)_\C \cong \Cst_\Roe(X)_\C\).
  Both isomorphisms are implemented by unitaries between the
  underlying Hilbert spaces.
\end{theorem}

\begin{proof}
  The proofs in the complex and real case are identical.  Let
  \(\pi\colon \Cont_0(X) \to \Cont_0(Y)\)
  be the restriction homomorphism.  Let
  \(\varrho_Y\colon \Cont_0(Y)\to\Bound(\Hils_Y)\)
  and \(\varrho_X\colon \Cont_0(X)\to\Bound(\Hils_X)\)
  be ample representations.  Then
  \(\varrho' \defeq \varrho_X \oplus \varrho_Y\circ\pi\)
  is an ample representation of~\(\Cont_0(X)\)
  on \(\Hils'\defeq \Hils_X \oplus \Hils_Y\).
  By Theorem~\ref{the:Roe_ample_unique}, we may use the particular
  representations \(\varrho'\)
  and~\(\varrho_Y\)
  to define \(\Cst_\Roe(X)\)
  and \(\Cst_\Roe(Y)\),
  because different ample representations give Roe
  \(\Cst\)\nb-algebras
  that are isomorphic through conjugation with a unitary between the
  underlying Hilbert spaces.

  Pick \(R>0\)
  such that for each \(x\in X\)
  there is \(y\in Y\)
  with \(d(x,y) \le R\).
  We build a Borel map \(g\colon X\to Y\)
  with \(g|_Y=\Id_Y\) and \(d(g(x),x)\le 2 R\) for all \(x\in X\),
  First, there is a countable cover \(X=\bigsqcup_{m\in\N} B_m\)
  by non-empty, disjoint Borel sets of diameter at most~\(R\)
  as in the proof of Theorem~\ref{the:Roe_ample_unique}.  For each
  \(m\in\N\),
  pick \(x_m\in B_m\)
  and \(y_m\in Y\)
  with \(d(x_m,y_m) \le R\).
  Define \(g(x) \defeq x\)
  for \(x\in Y\)
  and \(g(x) \defeq y_m\)
  for \(x\in B_m\setminus Y\),
  \(m\in\N\).  This map has all the required properties.

  The representation
  \(\varrho'\circ g^*\colon \Cont_0(Y)\to\Bound(\Hils')\)
  is ample because it contains
  \(\varrho_Y\circ \pi\circ g^* = \varrho_Y\)
  as a direct summand.  An operator on~\(\Hils'\)
  is locally compact or controlled for~\(\varrho'\)
  if and only if it is so for \(\varrho'\circ g^*\).
  Thus \(\Cst_\Roe(X,\varrho') = \Cst_\Roe(Y,\varrho'\circ g^*)\).
\end{proof}

In particular, Theorem~\ref{the:coarsely_dense} shows that all
coarsely dense subsets of~\(\R^d\)
have isomorphic Roe \(\Cst\)\nb-algebras.
This applies, in particular, to \(\Z^d\)
and to all Delone subsets of~\(\R^d\).
The latter are often used to model the atomic configurations of
materials that are not crystals (see, for instance,
\cite{Bellissard-Herrmann-Zarrouati:Hull}).  So all kinds of materials
lead to the same Roe \(\Cst\)\nb-algebra,
which depends only on the dimension~\(d\).

Theorem~\ref{the:coarsely_dense} suffices for our purposes, but we
mention that it extends to arbitrary coarse equivalences, see also
\cite{Higson-Roe:Analytic_K}*{Section 6.3}.

\begin{definition}
  \label{def:coarse_maps}
  Let \(X\)
  and~\(Y\)
  be proper metric spaces as above.  Two maps \(f_0,f_1\colon X\to Y\)
  are \emph{close} if there is \(R>0\)
  so that \(d(f_0(x),f_1(x))<R\)
  for all \(x\in X\).
  A \emph{coarse map} \(f\colon X\to Y\)
  is a Borel map with two properties: for any \(R>0\)
  there is \(S>0\)
  such that \(d(x,y)\le R\)
  for \(x,y\in X\)
  implies \(d(f(x),f(y))\le S\),
  and \(f^{-1}(B)\)
  is bounded in~\(X\)
  if \(B\subseteq Y\)
  is bounded.  A \emph{coarse equivalence} is a coarse map
  \(f\colon X\to Y\)
  for which there is another coarse map \(g\colon Y\to X\),
  called the \emph{coarse inverse} of~\(f\),
  such that \(g\circ f\)
  and \(f\circ g\)
  are close to the identity maps on \(X\)
  and~\(Y\),
  respectively.  We call \(X\)
  and~\(Y\)
  \emph{coarsely equivalent} if there is a coarse equivalence between
  them.
\end{definition}

For instance, the inclusion of a coarsely dense subspace is a coarse
equivalence: the proof of Theorem~\ref{the:coarsely_dense} builds a
coarse inverse for the inclusion map.

\begin{theorem}
  \label{the:coarse_equivalence_Roe}
  Let \(X\)
  and~\(Y\)
  be coarsely equivalent.  Then
  \(\Cst_\Roe(X)_\R \cong \Cst_\Roe(Y)_\R\)
  and \(\Cst_\Roe(X)_\C \cong \Cst_\Roe(Y)_\C\).
\end{theorem}

\begin{proof}
  Here it is more convenient to work with coarse spaces.  Let
  \(f\colon X\to Y\)
  be the coarse equivalence.  We claim that there is a coarse
  structure on the disjoint union \(X\sqcup Y\)
  such that both \(X\)
  and~\(Y\)
  are coarsely dense in \(X\sqcup Y\).
  This reduces the result to Theorem~\ref{the:coarsely_dense}.  We
  describe the desired coarse structure on \(X\sqcup Y\).
  A subset~\(E\)
  of \((X\sqcup Y)^2\)
  is called controlled if its intersections with \(X^2\)
  and~\(Y^2\)
  and the set of all \((f(x),y) \in Y^2\)
  for \((x,y)\in E\)
  or \((y,x)\in E\)
  are controlled.  This is a coarse structure on \(X\sqcup Y\)
  because~\(f\)
  is a coarse equivalence.  And the subspaces \(X\)
  and~\(Y\) are coarsely dense for the same reason.
\end{proof}

\subsection{Twists}
\label{sec:twists}

We show that magnetic twists do not change the isomorphism class of
the Roe \(\Cst\)\nb-algebra.
We let~\(X\)
be a \emph{discrete} metric space.  Let
\(\varrho\colon \Cont_0(X)\to \Bound(\Hils)\)
be a representation.  This is equivalent to a direct sum decomposition
\(\Hils = \bigoplus_{x\in X} \Hils_x\),
such that \(f\in \Cont_0(X)\)
acts by multiplication with \(f(x)\)
on the summand~\(\Hils_x\).
We assume for simplicity that each~\(\Hils_x\)
is non-zero.  This is weaker than being ample, which means that
each~\(\Hils_x\)
is infinite-dimensional.  So the following discussion also covers the
uniform Roe \(\Cst\)\nb-algebra of~\(X\).

We describe an operator on~\(\Hils\)
by a block matrix \((T_{x,y})_{x,y\in X}\)
with \(T_{x,y}\in \Bound(\Hils_y,\Hils_x)\).
These are multiplied by the usual formula,
\((S T)_{x,y} = \sum_{z\in X} S_{x,z} T_{z,y}\).
We twist this multiplication by a scalar-valued function
\(w\colon X\times X \times X\to \T\):
\[
(S *_w T)_{x,y} = \sum_{z\in X} w(x,z,y) S_{x,z} T_{z,y}.
\]
This defines a bounded bilinear map at least on the subalgebra
\(A(X,\varrho)\subseteq \Bound(\Hils)\)
of locally compact, controlled operators.

\begin{lemma}
  \label{lem:twisted_associative}
  The multiplication~\(*_w\)
  on \(A(X,\varrho)\)
  is associative if and only if
  \begin{equation}
    \label{eq:twist_cocycle_condition}
    w(x,z,y) w(x,t,z) = w(x,t,y) w(t,z,y)
  \end{equation}
  for all \(x,t,z,y\in X\).
\end{lemma}

\begin{proof}
  For \(S,T,U\in A(X,\varrho)\), we compute
  \begin{align*}
    \bigl((S *_w T) *_w U\bigr)_{x,y}
    &
      = \sum_{z,t\in X} w(x,z,y) w(x,t,z) S_{x,t} T_{t,z} U_{z,y},\\
    \bigl(S *_w (T *_w U)\bigr)_{x,y}
    &
      = \sum_{z,t\in X} w(x,t,y) w(t,z,y) S_{x,t} T_{t,z} U_{z,y}.
  \end{align*}
  The condition~\eqref{eq:twist_cocycle_condition} holds if and only if these are
  equal for all \(S,T,U\in A(X,\varrho)\)
  because all \(\Hils_x\) are non-zero.
\end{proof}

\begin{proposition}
  \label{pro:all_twists_trivial}
  If the function~\(w\)
  satisfies the cocycle condition in the previous lemma, then there is
  a function \(v\colon X\times X \to \T\) with
  \[
  w(x,z,y) =  v(x,z) v(z,y) v(x,y)^{-1}.
  \]
  The map \(\varphi\colon (A(X, \varrho),*_w) \to (A(X,\varrho),\cdot)\),
  \((T_{x,y})_{x,y\in X} \mapsto (v(x,y)\cdot T_{x,y})_{x,y\in X}\),
  is an algebra isomorphism.
\end{proposition}

\begin{proof}
  Fix a ``base point'' \(e\in X\)
  and let \(v(x,y) \defeq w(x,y,e)\).
  The condition~\eqref{eq:twist_cocycle_condition} for \((x,z,y,e)\)
  says that
  \[
  w(x,y,e) w(x,z,y) = w(x,z,e) w(z,y,e)
  \]
  holds for all \(x,z,y\in X\).  So
  \[
  v(x,z) v(z,y) v(x,y)^{-1}
  = w(x,z,e) w(z,y,e) w(x,y,e)^{-1}
  = w(x,z,y).
  \]
  The map~\(\varphi\)
  is a vector space isomorphism because \(v(x,y)\neq0\)
  for all \(x,y\in X\).  The computation
  \begin{multline*}
    \varphi (S *_w  T)_{x,y}
    = v(x,y) \sum_{z\in X} w(x,z,y) S_{x,z} T_{z,y}
    \\= \sum_{z\in X} w(x,z,y) v(x,y) v(x,z)^{-1} v(z,y)^{-1}
    \varphi(S)_{x,z} \varphi(T)_{z,y}
    = \sum_{z\in X} \varphi(S)_{x,z} \varphi(T)_{z,y}
  \end{multline*}
  shows that it is an algebra isomorphism.
\end{proof}

So the twisted and untwisted versions of \(A(X,\varrho)\)
are isomorphic algebras.  Thus
a magnetic field does not change the isomorphism type of the Roe
\(\Cst\)\nb-algebra.

Now let~\(X\)
be no longer discrete.  Then controlled, locally compact
operators are not given by matrices any more.  To write down the
twisted convolution as above, we use the smaller \Star{}algebra of
controlled, locally \emph{Hilbert--Schmidt} operators; it is
still dense in the Roe \(\Cst\)\nb-algebra.
Let us assume for simplicity that the representation~\(\varrho\)
for which we build the Roe \(\Cst\)\nb-algebra
has constant multiplicity, that is, it is the pointwise multiplication
representation on \(L^2(X,\mu) \otimes \Hils\)
for some regular Borel measure~\(\mu\)
on~\(X\)
and some Hilbert space~\(\Hils\).
Controlled, locally Hilbert--Schmidt operators on
\(L^2(X,\mu) \otimes \Hils\)
are the convolution operators for measurable functions
\(T\colon X\times X\to \ell^2(\Hils)\)
with controlled support and such that
\[
\int_{K\times K} \norm{T(x,y)}^2 \,\diff \mu(x)\,\diff\mu(y)  < \infty
\]
for all compact subsets \(K\subseteq X\)
and such that the resulting convolution operator is bounded.  The
multiplication of such operators is given by a convolution of their
integral kernels.  This may be twisted as above, using a Borel
function \(w\colon X^3\to\T\)
that satisfies the condition~\eqref{eq:twist_cocycle_condition}.  The
resulting function~\(v\)
in Proposition~\ref{pro:all_twists_trivial} is again Borel.  So the
isomorphism in Proposition~\ref{pro:all_twists_trivial} still works.
Thus the twist gives an isomorphic \Star{}algebra also in the
non-discrete case.

Following Bellissard~\cite{Bellissard:K-theory_solid}, a
\(d\)\nb-dimensional
material is often described through a crossed product
\(\Cst\)\nb-algebra
\(\Cont(\Omega)\rtimes\Z^d\)
for a compact space~\(\Omega\)
with a \(\Z^d\)\nb-action
by homeomorphisms and with an ergodic invariant measure
on~\(\Omega\).
To encode a magnetic field, the crossed product is replaced by the
crossed product \emph{twisted} by a \(2\)\nb-cocycle
\(\sigma\colon \Z^d \times \Z^d \to \Cont(\Omega,\T)\).
The space~\(\Omega\)
may be built as the ``hull'' of a point set or a fixed Hamiltonian,
see~\cite{Bellissard-Herrmann-Zarrouati:Hull}.  In this case, there is
a \emph{dense} orbit \(\Z^d\cdot\omega\)
in~\(\Omega\)
by construction.  So assuming the existence of a dense orbit is a
rather mild assumption in the context of Bellissard's theory.

We briefly explain why all twisted crossed products
\(\Cont(\Omega)\rtimes_\sigma \Z^d\)
as above are ``contained'' in the uniform Roe \(\Cst\)\nb-algebra
of~\(\Z^d\)
and hence also in the Roe \(\Cst\)\nb-algebra.
This observation is due to Kubota~\cite{Kubota:Controlled_bulk-edge}.

The main point here is the description of the uniform Roe
\(\Cst\)\nb-algebra
as a crossed product \(\ell^\infty(\Z^d)\rtimes\Z^d\),
see \cite{Roe:Lectures}*{Theorem~4.28}.  Let \(\omega\in\Omega\).
Then we define a \(\Z^d\)\nb-equivariant
\Star{}homomorphism
\(\epsilon_\omega\colon \Cont(\Omega) \to \ell^\infty(\Z^d)\)
by \((\epsilon_\omega f)(n) \defeq f(n\cdot \omega)\)
for all \(n\in\Z^d\),
\(f\in \Cont(\Omega)\).
This induces a \Star{}homomorphism
\(\Cont(\Omega)\rtimes_\sigma \Z^d \to \ell^\infty(\Z^d)
\rtimes_{\epsilon_\omega\circ \sigma} \Z^d\),
where~\(\rtimes_\sigma\)
denotes the crossed product twisted by a \(2\)\nb-cocyle~\(\sigma\).
The same argument that identifies the crossed product
\(\ell^\infty(\Z^d) \rtimes \Z^d\)
with the uniform Roe \(\Cst\)\nb-algebra
of~\(\Z^d\)
identifies
\(\ell^\infty(\Z^d) \rtimes_{\epsilon_\omega\circ \sigma} \Z^d\)
with a twist of the Roe \(\Cst\)\nb-algebra
as above.  Since all these twists give isomorphic \(\Cst\)\nb-algebras
by Proposition~\ref{pro:all_twists_trivial}, we get a
\Star{}homomorphism
\(\Cont(\Omega)\rtimes_\sigma \Z^d \to \ell^\infty(\Z^d) \rtimes
\Z^d\).
If the orbit of~\(\omega\)
is dense, then the \Star{}homomorphism~\(\epsilon_\omega\)
above is injective.  Then the induced \Star{}homomorphism
\(\Cont(\Omega)\rtimes_\sigma \Z^d \to \ell^\infty(\Z^d) \rtimes
\Z^d\)
is also injective.  Hence the uniform Roe \(\Cst\)\nb-algebra
really contains the twisted crossed product algebra.

Now we turn to the continuum version of the above theory.
Let~\(\Omega\)
be a compact space with a continuous action of~\(\R^d\).
This leads to crossed products \(\Cont(\Omega)\rtimes_\sigma \R^d\)
twisted, say, by Borel measurable \(2\)\nb-cocycles
\(\sigma\colon \R^d \times \R^d \to \Cont(\Omega,\T)\);
once again, the twist encodes a magnetic field.  Restricting to the
\(\R^d\)\nb-orbit
of some \(\omega\in\Omega\)
maps \(\Cont(\Omega)\rtimes_\sigma \R^d\)
to
\(\Contb^\mathrm{u}(\R^d)\rtimes_{\epsilon_\omega\circ\sigma} \R^d\)
for the \(\Cst\)\nb-algebra
\(\Contb^\mathrm{u}(\R^d)\)
of bounded, uniformly continuous functions on~\(\R^d\).
We have seen in Example~\ref{exa:Roe_Rn} that
\(\Contb^\mathrm{u}(\R^d)\rtimes \R^d \subseteq L^\infty(\R^d) \cdot
\Cst(\R^d)\)
is contained in the Roe \(\Cst\)\nb-algebra
of~\(\R^d\).
This remains the case also in the twisted case because a Borel
measurable \(2\)\nb-cocycle
\(\sigma\colon \R^d \times \R^d \to \Cont(\Omega,\T)\)
defines a Borel function \((\R^d)^3 \to\T\),
which is untwisted by Proposition~\ref{pro:all_twists_trivial}.  So
all twisted crossed products \(\Cont(\Omega)\rtimes_\sigma \R^d\)
map to the Roe \(\Cst\)\nb-algebra
of~\(\R^d\).
As above, this map is an embedding if the orbit of~\(\omega\)
is dense in~\(\Omega\).

The Roe \(\Cst\)\nb-algebras
for \(\R^d\)
and~\(\Z^d\)
are isomorphic by Theorem~\ref{the:coarsely_dense}.  So there is a unique
Roe \(\Cst\)\nb-algebra
in each dimension that contains all the twisted crossed product
algebras that are used as models for disordered materials, both in
continuum models and tight binding models.  This fits interpreting
the twisted crossed products \(\Cont(\Omega)\rtimes_\sigma \Z^d\)
or \(\Cont(\Omega)\rtimes_\sigma \R^d\)
as models for disorder with built-in \emph{a priori} restrictions,
whereas the Roe \(\Cst\)\nb-algebra describes general disorder.

\subsection{Approximation by controlled operators as a continuity property}
\label{sec:approximation_continuity}

In order to belong to the Roe \(\Cst\)\nb-algebra,
an operator has to be a norm limit of locally compact, controlled
operators.  Any such norm limit is again locally compact.  The
property of being a norm limit of controlled operators may be hard to
check.  A tool for this is Property~A, an approximation property for
coarse spaces that ensures that elements of the Roe
\(\Cst\)\nb-algebra
may be approximated in a systematic way by controlled operators, see
\cites{Roe:Warped_A, Brodzki-Cave-Li:Exactness}.  We also mention the
related Operator Norm Localization Property for subspaces of~\(\R^d\).

We now specialise to the case where~\(X\)
is a closed subset of~\(\R^d\)
with the restriction of the Euclidean metric.  Such spaces have
Property~A.  We use it to define complex Roe \(\Cst\)\nb-algebras
through continuity for a certain representation of~\(\R^d\).
We fix a representation \(\varrho\colon \Cont_0(X) \to \Bound(\Hils)\)
on a complex Hilbert space~\(\Hils\).
Let \(\bar\varrho\colon \Contb(X)\to\Bound(\Hils)\)
be its unique strictly continuous extension to the multiplier algebra.
For \(t\in\R^d\),
define \(\Euler_t\in\Contb(X)\)
by \(\Euler_t(x) \defeq \Euler^{\ima x\cdot t}\).
The map \(t\mapsto \Euler_t\)
is continuous for the strict topology on~\(\Contb(X)\).
Hence the representation~\(\sigma\)
of~\(\R^d\)
on~\(\Hils\)
defined by \(\sigma_t(\xi) \defeq \bar\varrho(\Euler_t)(\xi)\)
is continuous.  This representation is generated by the position
operators.  If \(X\subseteq \Z^d\),
then \(\Euler_t=1\)
for \(t\in 2\pi\Z^d\),
so that the representation~\(\sigma\)
descends to the torus \((\R/2\pi \Z)^d\).

By conjugation, \(\sigma\)
induces an action \(\Ad \sigma\)
of~\(\R^d\)
by automorphisms of \(\Bound(\Hils)\).
We call \(S\in\Bound(\Hils)\)
\emph{continuous} with respect to \(\Ad \sigma\)
if the map \(\R^d \to \Bound(\Hils)\),
\(t\mapsto \Ad \sigma_t(S)\),
is continuous in the norm topology on~\(\Bound(\Hils)\).
The following theorem describes the Roe \(\Cst\)\nb-algebra
through this continuity property:

\begin{theorem}
  \label{the:controlled_continuity}
  An operator \(S\in\Bound(\Hils)\)
  is a norm limit of controlled operators if and only if it is
  continuous with respect to \(\Ad \sigma\).
  And \(\Cst_\Roe(X,\varrho)\)
  is the \(\Cst\)\nb-subalgebra
  of all operators on~\(\Hils\)
  that are locally compact and continuous with respect to
  \(\Ad \sigma\).
\end{theorem}

\begin{proof}
  For each \(S\in\Bound(\Hils)\),
  the map \(\R^d\to\Bound(\Hils)\)
  is continuous for the strong topology on \(\Bound(\Hils)\).
  Therefore, the \(\Bound(\Hils)\)-valued integral
  \[
  f*S \defeq \int_{\R^d} f(t) \Ad \sigma_t(S) \,\diff t
  \]
  makes sense for any \(f\in L^1(\R^d)\).
  Let \((f_n)_{n\in\N}\)
  be a bounded approximate unit in the Banach algebra \(L^1(\R^d)\).
  We claim that~\(S\)
  is continuous if and only if \((f_n*S)_{n\in\N}\)
  converges in the norm topology to~\(S\).
  It is well known that any continuous representation of~\(\R^d\)
  becomes a nondegenerate module over \(L^1(\R^d)\).
  Thus \((f_n*S)_{n\in\N}\)
  converges in norm to~\(S\)
  if~\(S\)
  is continuous.  Conversely, operators of the form \(f*S\)
  are continuous because the action of~\(\R^d\)
  on \(L^1(\R^d)\)
  is continuous.  Since the set of continuous operators is closed in
  the norm topology, \(S\)
  is continuous if \((f_n*S)_{n\in\N}\) converges in norm to~\(S\).

  There is an approximate unit~\((f_n)_{n\in\N}\)
  for~\(L^1(\R^d)\)
  such that the Fourier transform of each~\(f_n\)
  has compact support.  For instance, we may use the Fejér kernel
  \[
  \Lambda(x_1,\dotsc,x_n) \defeq \prod_{j=1}^d \frac{\sin^2(\pi x_j)}{\pi^2 x_j^2},
  \]
  which has Fourier transform \(\prod_{j=1}^d (1-\abs{x_j})_+\),
  and rescale it to produce an approximate unit for \(L^1(\R^d)\).
  We claim that \(f_n*S\)
  is controlled for each \(n\in\N\).
  More precisely, assume that \(\widehat{f_n}\)
  is supported in the ball of radius~\(R\)
  in~\(\R^d\).
  We claim that~\(f_n*S\)
  is \(R\)\nb-controlled.

  This is easy to prove if~\(X\)
  is discrete.  Then we may describe operators on~\(\Hils\)
  using matrix coefficients \(S_{x,y}\in\Bound(\Hils_y,\Hils_x)\)
  for \(x,y\in X\).
  A direct computation shows that the matrix coefficients of \(f_n*S\)
  are \(\widehat{f_n}(x-y)\cdot S_{x,y}\).
  This vanishes for \(\norm{x-y}>R\).
  So~\(f_n*S\)
  is \(R\)\nb-controlled
  as asserted.  The following argument extends this result to the case
  where~\(X\) is not discrete, such as \(X=\R^d\).

  Let \(U,V\subseteq \R^d\)
  be two relatively compact, open subsets of distance at least~\(R\)
  and let \(g,h\in\Cont^\infty(\R^d)\)
  be smooth functions supported in \(U\)
  and~\(V\), respectively.  Then
  \begin{equation}
    \label{eq:Roe_through_continuity_integral}
    \int_{\R^d} g(y) \Euler^{\ima y\cdot t}\cdot h(x) \Euler^{-\ima x\cdot t} f_n(t) \,\diff t
    = g(y) h(x) \widehat{f_n}(y-x)
    = 0
  \end{equation}
  for all \(x,y\in \R^d\).  We restrict \(g,h\) to~\(X\) and compute
  \begin{align*}
    \varrho(g) (f_n* S) \varrho(h)
    &\defeq \int_{\R^d} \varrho(g) \bar\varrho(\Euler_t) S
    \bar\varrho(\Euler_{-t}) \varrho(h) f_n(t) \,\diff t
    \\&= \int_{\R^d} \varrho(g\cdot \Euler_t) S
    \varrho(h\cdot \Euler_{-t}) f_n(t) \,\diff t.
  \end{align*}
  Let \(\Cont_0^\infty(U\times V)\)
  denote the Fréchet space of smooth function on~\(\R^{2d}\)
  supported in \(U\times V\).
  We may identify this with the complete projective tensor product of
  \(\Cont^\infty_0(U)\)
  and \(\Cont^\infty_0(V)\).
  Hence there is a continuous linear map
  \begin{equation}
    \label{eq:function_on_both_side}
    \Cont^\infty_0(U\times V)\to \Bound(\Hils),\qquad
    g\otimes h\mapsto \varrho(g) S \varrho(h).
  \end{equation}
  The integral in~\eqref{eq:Roe_through_continuity_integral} converges
  to~\(0\)
  in the Fréchet topology of \(\Cont^\infty_0(U\times V)\).
  Hence~\eqref{eq:Roe_through_continuity_integral} implies
  \(\varrho(g) (f_n*S) \varrho(h)=0\).

  The continuous map~\eqref{eq:function_on_both_side} still exists if
  \(U\)
  and~\(V\)
  are not of distance~\(R\).
  If~\(S\)
  is locally compact, then \(\varrho(g) S \varrho(h)\)
  is a compact operator on~\(\Hils\)
  for all \(g\in\Cont^\infty_0(U)\),
  \(h\in\Cont^\infty_0(V)\).
  This remains so for all operators in the image
  of~\eqref{eq:function_on_both_side} by continuity.  Therefore,
  \(\varrho(g) (f_n*S)\varrho(h)\)
  is compact for all \(g,h\)
  as above.  Choosing~\(U\)
  large enough, we may take~\(g\)
  to be constant equal to~\(1\)
  on the \(R\)\nb-neighbourhood
  of~\(V\).
  Then \(\varrho(g) (f_n*S)\varrho(h) =(f_n*S)\varrho(h)\)
  because~\(f_n*S\)
  is \(R\)\nb-controlled.
  So operators of the form \((f_n*S)\varrho(h)\)
  with smooth, compactly supported~\(h\)
  are compact.  Since any continuous, compactly supported function is
  dominated by a smooth, compactly supported function, we get the same
  for all \(h\in\Contc(X)\).
  A similar argument shows that \(\varrho(g)(f_n*S)\)
  is compact for all \(g\in\Contc(X)\).
  Hence the operators \(f_n*S\)
  are locally compact, controlled operators if~\(S\)
  is locally compact.
\end{proof}

Property~A is equivalent to the ``Operator Norm Localization
Property'' for metric spaces with bounded geometry,
see~\cite{Sako:A_operator_localization}.  Roughly speaking, this
property says that the operator norm of a controlled operator may be
computed using vectors in the Hilbert space with bounded support.  The
support of a vector \(\xi\in\Hils\)
is the set of all \(x\in X\)
such that \(f\cdot \xi\neq0\)
for all \(f\in\Cont_0(X)\)
with \(f(x)\neq0\).
We formulate this property for subspaces of~\(\R^d\):

\begin{theorem}
  \label{the:ONL}
  Let \(X\subseteq \R^d\)
  and let \(\varrho\colon \Cont_0(X)\to\Bound(\Hils)\)
  be a representation.  Pick scalars \(R>0\)
  and \(c\in(0,1)\).
  Then there is a scalar \(S>0\)
  such that for any \(R\)\nb-controlled
  operator \(T\in\Bound(\Hils)\),
  there is \(\xi\in\Hils\)
  with \(\norm{\xi}=1\)
  such that the support of~\(\xi\)
  has diameter at most~\(S\)
  and \(\norm{T(\xi)} \ge \norm{T} \ge c\cdot \norm{T(\xi)}\).
\end{theorem}

\begin{proof}
  The statement of the theorem is that the space~\(X\) has the
  ``Operator Norm Localisation Property'' defined
  in~\cite{Chen-Tessera-Wang:Operator_localization}.  This property
  is invariant under coarse equivalence and passes to subspaces by
  \cite{Chen-Tessera-Wang:Operator_localization}*{Propositions 2.5
    and~2.6}.
  \cite{Chen-Tessera-Wang:Operator_localization}*{Theorem~3.11 and
    Proposition~4.1} show that solvable Lie groups such as~\(\R^d\)
  have this property, and hence also all subspaces of~\(\R^d\).
\end{proof}

\subsection{Dense subalgebras with isomorphic K-theory}
\label{sec:dense_same_K}

Let~\(A\)
be a \(\Cst\)\nb-algebra
with a continuous \(\R^d\)\nb-action
\(\alpha\colon \R^d\to\Aut(A)\).
The action defines several canonical \Star{}subalgebras of~\(A\)
with the same \(\K\)\nb-theory.
The \Star{}subalgebra of \emph{smooth elements} is
\[
A^\infty \defeq \setgiven{a\in A}
{t\mapsto \alpha_t(a) \text{ is a smooth function } \R^d\to A}.
\]
This Fréchet \Star{}subalgebra is closed under holomorphic functional
calculus and also under smooth functional calculus for normal
elements, see~\cite{Blackadar-Cuntz:Differential}.

Let \(F \subseteq \R^d\)
be a compact convex subset with non-empty interior and
containing~\(0\).
Let \(\mathcal{O}(A,\alpha,F) \subseteq A\)
be the set of all \(a\in A\)
for which the function \(\R^d\ni t\mapsto \alpha_t(a)\)
extends to a continuous function on \(\R^d + \ima F\)
that is holomorphic on the interior of \(\R^d + \ima F\).
This is a dense Banach subalgebra in~\(A\),
and the inclusion \(\mathcal{O}(A,\alpha,F) \injto A\)
induces an isomorphism on topological \(\K\)\nb-theory
by \cite{Bost:Principe_Oka}*{Théorème~2.2.1}.  Let
\(\mathcal{O}^\infty(A,\alpha,F) \subseteq A\)
be the set of those \(a\in A\)
for which the function \(\R^d\ni t\mapsto \alpha_t(a)\)
extends to a smooth function on \(\R^d + \ima F\)
that is holomorphic on the interior of \(\R^d + \ima F\).
The inclusion \(\mathcal{O}^\infty(A,\alpha,F) \injto A\)
induces an isomorphism on topological \(\K\)\nb-theory
as well.  If \(F_1\subseteq F_2\),
then \(\mathcal{O}(A,\alpha,F_2) \injto \mathcal{O}(A,\alpha,F_1)\).
There are two important limiting cases of the subalgebras
\(\mathcal{O}(A,\alpha,F)\).

First, let~\(F\)
run through a neighbourhood basis of~\(0\)
in~\(\R^d\).
Then the dense Banach subalgebras \(\mathcal{O}(A,\alpha,F)\)
form an inductive system, whose colimit is the dense \Star{}subalgebra
\(A^\omega \subseteq A\)
of all \emph{real-analytic elements} of~\(A\),
that is, those \(a\in A\)
with the property that each \(t\in\R^d\)
has a neighbourhood on which \(s\mapsto \alpha_s(a)\)
is given by a convergent power series with coefficients in~\(A\).
The subalgebra~\(A^\omega\)
is still closed under holomorphic functional calculus by
\cite{Meyer:HLHA}*{Proposition~3.46}.  This gives an easier
explanation than Bost's Oka principle why~\(A^\omega\)
has the same topological \(\K\)\nb-theory as~\(A\).

Secondly, let~\(F\)
run through an increasing sequence whose union is~\(\R^d\).
Then the dense Banach subalgebras \(\mathcal{O}(A,\alpha,F)\)
form a projective system, whose limit is the dense \Star{}subalgebra
\(\mathcal{O}(A,\alpha)\)
of all \emph{holomorphic elements} \(a\in A\),
that is, those elements for which the map
\(\R^d \ni t\mapsto \alpha_t(a)\)
extends to a holomorphic function on~\(\C^d\).
This is a locally multiplicatively convex Fréchet algebra.
Phillips~\cite{Phillips:K_Frechet} has extended
topological \(\K\)\nb-theory
to such algebras.  The Milnor \(\varprojlim^1\)-sequence
in \cite{Phillips:K_Frechet}*{Theorem 6.5} shows that the
inclusion \(\mathcal{O}(A,\alpha)\injto A\)
induces an isomorphism in topological \(\K\)\nb-theory.

We apply all this to the Roe \(\Cst\)\nb-algebra
of~\(\Z^d\)
and the continuous \(\R^d\)\nb-action~\(\sigma\)
defined in Section~\ref{sec:approximation_continuity}.  Here this
action descends to the torus~\(\T^d\),
which simplifies the study of the dense subalgebras above.  We
describe the dense subalgebras of smooth, real-analytic and
holomorphic elements in \(\Cst_\Roe(\Z^d)\).
All these have the same topological \(\K\)\nb-theory.
Let \(\varrho\colon \Cont_0(\Z^d)\to\Bound(\Hils)\)
be a representation on a separable Hilbert space, not necessarily
ample.  Let~\(\Hils_x\)
for \(x\in X\)
be the fibres of~\(\Hils\)
with respect to~\(\varrho\).
Describe operators on~\(\Hils\)
by block matrices \((T_{x,y})_{x,y\in\Z^d}\)
with \(T_{x,y}\in\Bound(\Hils_y,\Hils_x)\)
for all \(x,y\in\Z^d\).  Then
\[
\sigma_t(T_{x,y}) = (\exp(\ima t\cdot (x-y)) T_{x,y})_{x,y\in\Z^d}.
\]

\begin{proposition}
  \label{pro:smooth_analytic_in_Roe_Zd}
  A block matrix \((T_{x,y})_{x,y\in\Z^d}\)
  as above gives a smooth element for the
  \(\R^d\)\nb-action~\(\sigma\) on~\(\Cst_\Roe(\Z^d)\) if and only if the function
  \[
  \Z^d \ni k \mapsto \sup_{n\in\Z^d} \{ \norm{T_{n,n+k}} \}
  \]
  has rapid decay, that is, for each \(a>0\)
  there is a constant~\(C_a>0\)
  such that \(\norm{T_{n,n+k}} \le C_a (1+ \norm{k})^{-a}\)
  for all \(n,k\in\Z^d\).
  It gives a real-analytic element for~\(\sigma\)
  if and only if there are \(a>0\) and \(C_a>0\) such that
  \[
  \norm{T_{n,n+k}} \le C_a \cdot \exp(-a \norm{k})
  \]
  for all \(n,k\in\Z^d\).
  It gives a holomorphic element for~\(\sigma\)
  if and only if for each \(a>0\) there is \(C_a>0\) such that
  \[
  \norm{T_{n,n+k}} \le C_a \cdot \exp(-a \norm{k}).
  \]
\end{proposition}

\begin{proof}
  The \(j\)th
  generator of the \(\R^d\)\nb-action~\(\sigma\)
  maps a block matrix~\((T_{x,y})_{x,y\in\Z^d}\)
  to
  \[
  \lim_{t\to0} \frac{1}{t} (\sigma_{t e_j}(T_{x,y}) - (T_{x,y}))_{x,y\in\Z^d}
  = ((x_j-y_j)T_{x,y})_{x,y\in\Z^d}.
  \]
  Hence polynomials in these generators multiply the
  entries~\(T_{x,y}\)
  with polynomials in \(x-y\in\Z^d\).
  So~\((T_{x,y})_{x,y\in\Z^d}\)
  belongs to a smooth element of~\(\Cst_\Roe(\Z^d)\)
  if and only if \((p(x-y)\cdot T_{x,y})_{x,y\in\Z^d}\)
  belongs to a bounded operator for each polynomial~\(p\)
  in \(d\)~variables.
  It suffices to consider the polynomials \(1+\norm{x-y}_2^{2 b}\)
  for \(b\in\N\).
  Since the operator norm for diagonal block matrices is the supremum
  of the operator norms of the entries, we see that the boundedness of
  \(((1+ \norm{x-y}_2^{2 b})\cdot T_{x,y})_{x,y\in\Z^d}\)
  for all \(b\in\N\)
  is equivalent to the boundedness of
  \(\sup_{k,n\in\Z^d} \norm{T_{n,n+k}} (1+ \norm{k}_2^{2 b})\)
  for all \(b\in\N\).
  This proves the claim about the smooth elements.  The analytic
  extension of~\(\sigma\)
  to \(\ima z\in \C^d\)
  must map~\((T_{x,y})_{x,y\in\Z^d}\)
  to \(((\exp(z\cdot (x-y))T_{x,y})_{x,y\in\Z^d}\).
  Thus~\((T_{x,y})_{x,y\in\Z^d}\)
  describes an element of
  \(\mathcal{O}^\infty(\Cst_\Roe(\Z^d),\sigma,F)\) if and only if
  \begin{equation}
    \label{eq:O_F_estimate}
    \sup_{k,n\in\Z^d} \norm{T_{n,n+k}} (1+ \norm{k}_2^{2 b}) \exp(z\cdot k)<\infty    
  \end{equation}
  for all \(z\in F\),
  \(b\in\N\).
  When we let \(F\searrow\{0\}\)
  or \(F\nearrow \C^d\),
  we may leave out the polynomial factors because they are dominated
  by \(\exp(z\cdot k)\).
  This proves the claims about the
  real-analytic elements and holomorphic elements.
\end{proof}

Estimates of the form
\(\norm{T_{x,y}} \le C_a \exp(-a\cdot \norm{x-y})\)
for some \(a>0\),
\(C_a>0\)
play an important role in the study of Anderson localisation; see, for
instance,
\cite{Aizenman-Molchanov:Localization_elementary}*{Equation~(2.3)}.

\subsection{Approximate unit of projections}
\label{sec:apprid_projection}

Unlike the uniform Roe \(\Cst\)\nb-algebra,
the Roe \(\Cst\)\nb-algebra
of a proper metric space is never unital.  Instead, it has an
approximate unit of projections:

\begin{proposition}
  \label{pro:Roe_apprid}
  Let~\(X\)
  be a proper metric space and let
  \(\varrho\colon \Cont_0(X)\to\Bound(\Hils)\)
  be a representation.  The Roe \(\Cst\)\nb-algebra
  \(\Cst_\Roe(X,\varrho)\) has an approximate unit of projections.
\end{proposition}

\begin{proof}
  Any proper metric space contains a coarsely dense, discrete
  subspace.  By Theorem~\ref{the:coarsely_dense}, we may assume
  that~\(X\)
  itself is discrete.  By Theorem~\ref{the:Roe_ample_unique}, we may
  further assume that the Roe \(\Cst\)\nb-algebra
  is built using the obvious representation of~\(\Cont_0(X)\)
  on \(\ell^2(X,\ell^2(\N))\).
  Then the Roe \(\Cst\)\nb-algebra
  contains \(\ell^\infty(X, \Cont_0(\N))\)
  as multiplication operators.  Any function \(h\colon X\to\N\)
  defines a projection in \(\ell^\infty(X, \Cont_0(\N))\),
  namely, the characteristic function of
  \(\setgiven{(x,n)\in X\times\N}{n<h(x)}\).
  These projections form an approximate unit in the Roe
  \(\Cst\)\nb-algebra.
\end{proof}

Let \((p_\alpha)_{\alpha \in S}\)
be an approximate unit of projections in \(\Cst_\Roe(X,\varrho)\).
Then \(\Cst_\Roe(X,\varrho)\)
is isomorphic to the inductive limit
\[
\Cst_\Roe(X,\varrho) = \varinjlim p_\alpha \Cst_\Roe(X,\varrho) p_\alpha.
\]
Since \(\K\)\nb-theory commutes with inductive limits, we get
\begin{equation}
  \label{eq:K_Roe_indlim}
  \K_*(\Cst_\Roe(X,\varrho)) \cong
  \varinjlim \K_*(p_\alpha \Cst_\Roe(X,\varrho) p_\alpha).
\end{equation}
Each of the corners \(p_\alpha \Cst_\Roe(X,\varrho) p_\alpha\)
is unital.  This simplifies the definition of the groups
\(\K_*(p_\alpha \Cst_\Roe(X,\varrho) p_\alpha)\)
for fixed~\(\alpha\).

\begin{corollary}
  \label{cor:K0_Roe_projections}
  Any class in \(\K_0(\Cst_\Roe(X))\)
  is represented by a formal differences of projections in
  \(\Cst_\Roe(X)\).
  Two such differences \([p_+]-[p_-]\)
  and \([q_+]-[q_-]\)
  represent the same class in \(\K_0(\Cst_\Roe(X))\)
  if and only if there is a projection~\(r\)
  in \(\Cst_\Roe(X)\)
  such that the projections \(p_+ \oplus q_- \oplus r\)
  and \(p_- \oplus q_+ \oplus r\)
  in \(\Mat_3(\Cst_\Roe(X))\) are Murray--von Neumann equivalent.
\end{corollary}

\begin{proof}
  Any class in \(\K_0(\Cst_\Roe(X))\)
  is the image of a class in \(\K_0(p_\alpha \Cst_\Roe(X) p_\alpha)\)
  for some~\(\alpha\)
  by~\eqref{eq:K_Roe_indlim}.  Since
  \(p_\alpha \Cst_\Roe(X) p_\alpha\)
  is unital, it is represented by a formal difference of two
  projections in \(\Mat_k(p_\alpha \Cst_\Roe(X) p_\alpha)\)
  for some \(k\in\N_{\ge1}\).
  These projections also belong to \(\Mat_k(\Cst_\Roe(X))\),
  which is isomorphic to~\(\Cst_\Roe(X)\)
  by Corollary~\ref{cor:Roe_matrix-stable}.  Even more, the
  isomorphism \(\Mat_k(\Cst_\Roe(X))\cong\Cst_\Roe(X)\)
  is by conjugation with a unitary \(k\times1\)-matrix
  over the multiplier algebra of~\(\Cst_\Roe(X)\).
  Hence any projection in \(\Mat_k(\Cst_\Roe(X))\)
  is Murray--von Neumann equivalent to one in~\(\Cst_\Roe(X)\).
  Thus any class in \(\K_0(\Cst_\Roe(X))\)
  is represented by a formal difference of projections in
  \(\Cst_\Roe(X)\).
  Using~\eqref{eq:K_Roe_indlim} and the definition of~\(\K_0\)
  for unital algebras once again, we see that
  \([p_+]-[p_-]=[q_+]-[q_-]\)
  holds if and only if there is a projection~\(r\)
  in \(\Mat_k(p_\alpha \Cst_\Roe(X) p_\alpha)\)
  for some \(k,\alpha\)
  such that \(p_+ \oplus q_- \oplus r\)
  and \(p_- \oplus q_+ \oplus r\)
  are Murray--von Neumann equivalent in
  \(\Mat_{k+2}(p_\alpha \Cst_\Roe(X) p_\alpha)\).
  Here we may replace~\(r\)
  by an equivalent projection and so reduce to \(k=1\).
  And since any projection in~\(\Cst_\Roe(X)\)
  is equivalent to one in \(p_\alpha \Cst_\Roe(X) p_\alpha\)
  for some~\(\alpha\),
  we may allow~\(r\) to be any projection in~\(\Cst_\Roe(X)\).
\end{proof}

\section{Coarse Mayer--Vietoris sequence}
\label{sec:coarse_MV}

We now recall how to compute the \(\K\)\nb-theory
of the Roe \(\Cst\)\nb-algebra
of \(\R^d\)
or~\(\Z^d\)
using the coarse Mayer--Vietoris sequence.  This method is due to
Higson--Roe--Yu~\cite{Higson-Roe-Yu:Coarse_Mayer-Vietoris} and was
also explained by Kubota in
\cite{Kubota:Controlled_bulk-edge}*{Section 2.3.3}.  The Roe
\(\Cst\)\nb-algebra
is defined also for the half-space
\(\Z^{d-1}\times\N \subseteq \Z^d\),
and its \(\K\)\nb-theory
vanishes.  We identify the Roe \(\Cst\)\nb-algebra
of a subspace with a corner in the larger Roe \(\Cst\)\nb-algebra
and describe the ideal generated by this corner as a relative Roe
\(\Cst\)\nb-algebra.
Then we recall the Mayer--Vietoris sequence for two ideals in a
\(\Cst\)\nb-algebra.
We give an elegant proof due to Wodzicki, which reduces its exactness
to ordinary long exact sequences for \(\Cst\)\nb-algebra
extensions.  This also describes the boundary map in the
Mayer--Vietoris sequence through the boundary map for a
\(\Cst\)\nb-algebra
extension.  Hence the many tools and formulas for the boundary maps
for extensions of \(\Cst\)\nb-algebras
also apply to the boundary map in the Mayer--Vietoris sequence.

\subsection{Subspaces and corners}
\label{sec:subspace_corner}

Theorem~\ref{the:coarsely_dense} shows that \(\Cst_\Roe(Y)\)
and \(\Cst_\Roe(X)\)
are isomorphic if \(Y\subseteq X\)
is coarsely dense.  Now let \(Y\subseteq X\)
be an arbitrary closed subset, still with the restriction of the
metric from~\(X\).
We are going to relate the Roe \(\Cst\)\nb-algebras
of \(Y\)
and~\(X\).
First, we are going to show that \(\Cst_\Roe(Y)\)
is isomorphic to a corner in \(\Cst_\Roe(X)\);
this is an easy consequence of Theorem~\ref{the:Roe_ample_unique}.
For the coarse Mayer--Vietoris sequence, we also need to know that the
ideal generated by this corner is the relative Roe
\(\Cst\)\nb-algebra, defined as follows:

\begin{definition}
  \label{def:Roe_relative}
  Let \(Y\subseteq X\)
  and let \(\varrho\colon \Cont_0(X)\to\Bound(\Hils)\)
  be a representation.  An operator \(T\in\Bound(\Hils)\)
  is \emph{supported near}~\(Y\) if there is \(R>0\) with
  \[
  \supp(T) \subseteq \setgiven{(x,y)\in X \times X}
  {d(x,Y)<R\text{ and }d(y,Y)<R}.
  \]
  The \emph{relative Roe algebra} \(\Cst_\Roe(Y\subseteq X,\varrho)\)
  is the \(\Cst\)\nb-subalgebra
  of~\(\Bound(\Hils)\)
  generated by the controlled locally compact operators supported
  near~\(Y\).
\end{definition}

If \(T\)
is supported near~\(Y\)
and~\(S\)
is controlled, then \(S T\)
and \(T S\)
are again supported near~\(Y\),
and~\(T^*\)
is also supported near~\(Y\).
Thus the controlled operators supported near~\(Y\)
form a (two-sided) \Star{}ideal in the \Star{}algebra of controlled
operators.  Hence \(\Cst_\Roe(Y\subseteq X,\varrho)\)
is the closure of this \Star{}algebra in \(\Cst_\Roe(X,\varrho)\),
and it is a closed two-sided \Star{}ideal in \(\Cst_\Roe(X,\varrho)\).

A \emph{corner} in a \(\Cst\)\nb-algebra~\(A\)
is a \(\Cst\)\nb-subalgebra
of the form \(P A P\)
for a projection~\(P\)
in the multiplier algebra of~\(A\).
Any corner is a hereditary subalgebra.  It is canonically Morita
equivalent to the ideal generated by~\(P\),
which we denote by~\(A P A\)
because it is the closed linear span of~\(a_1 P a_2\)
for \(a_1,a_2\in A\).
The imprimitivity bimodule is~\(A P\)
with the obvious full Hilbert \(A P A, P A P\)-bimodule
structure, obtained by restricting the usual Hilbert \(A,A\)-bimodule
structure on~\(A\).

\begin{theorem}
  \label{the:subset_full_corner}
  Let \(Y\subseteq X\)
  be a closed subspace with the subspace metric.  Then
  \(\Cst_\Roe(X)\)
  is isomorphic to a corner in \(\Cst_\Roe(Y)\).
  The ideal generated by it is the relative Roe \(\Cst\)\nb-algebra
  \(\Cst_\Roe(Y\subseteq X)\).
  This holds both in the complex and real case.
\end{theorem}

\begin{proof}
  We prove the real case.  The complex case is analogous.  We choose
  ample representations \(\varrho_X\)
  and~\(\varrho_Y\)
  of \(\Cont_0(X)\)
  and \(\Cont_0(Y)\)
  on separable real Hilbert spaces \(\Hils_X\)
  and~\(\Hils_Y\),
  respectively.  Using the restriction map
  \(p\colon \Cont_0(X)\to\Cont_0(Y)\),
  we build another ample representation
  \(\varrho'\defeq \varrho_Y\circ p\oplus \varrho_X\)
  of~\(\Cont_0(X)\)
  on the separable real Hilbert space \(\Hils_Y\oplus\Hils_X\).
  We use~\(\varrho'\)
  to build \(\Cst_\Roe(X)\),
  which is allowed by Theorem~\ref{the:Roe_ample_unique}.  The
  projection~\(P\)
  onto the summand \(\Hils_Y\)
  is a multiplier of \(\Cst_\Roe(X)\)
  because it is \(0\)\nb-controlled.
  (It is only a multiplier because it is not locally compact.)  Since
  the map \(x\mapsto P x P\)
  on \(\Bound(\Hils_Y\oplus\Hils_X)\)
  is bounded, the corner \(P \Cst_\Roe(X) P\)
  is the closure of the space of all controlled, locally compact
  operators on~\(\Hils_Y\).
  Here we should use the representation \(\varrho_Y\circ p\)
  of~\(\Cont_0(X)\)
  to define controlled operators and local compactness.  But
  since~\(Y\)
  carries the subspace metric from~\(X\)
  and~\(p\)
  is surjective, the representations \(\varrho_Y\circ p\)
  of~\(\Cont_0(X)\)
  and~\(\varrho_Y\)
  of~\(\Cont_0(Y)\)
  define the same controlled or locally compact operators.  Hence
  \(P \Cst_\Roe(X,\varrho') P = \Cst_\Roe(Y,\varrho_Y)\).
  So~\(\Cst_\Roe(Y)\)
  is isomorphic to a corner in~\(\Cst_\Roe(X)\).

  The corner \(P \Cst_\Roe(X) P\)
  is contained in \(\Cst_\Roe(Y\subseteq X)\).
  Since the latter is an ideal, the ideal
  \(\Cst_\Roe(X) P \Cst_\Roe(X)\)
  is also contained in \(\Cst_\Roe(Y\subseteq X)\).
  For the converse inclusion, we must show that any controlled,
  locally compact operator~\(T\)
  that is supported near~\(Y\)
  belongs to \(\Cst_\Roe(X) P \Cst_\Roe(X)\).
  Let \(R>0\)
  and let~\(T\)
  be supported in the \(R\)\nb-neighbourhood
  of~\(Y\).
  This \(R\)\nb-neighbourhood
  is a closed subspace~\(Y_R\)
  of~\(X\),
  and \(Y\subseteq Y_R\)
  is coarsely dense by construction.  The restriction of~\(\varrho'\)
  to the Hilbert subspace
  \(\Hils_{Y,R} \defeq \varrho'(1_{Y_R})(\Hils_Y\oplus \Hils_X)\)
  is an ample representation of~\(\Cont_0(Y_R)\).
  Hence it defines \(\Cst_\Roe(Y_R)\)
  by Theorem~\ref{the:Roe_ample_unique}.  This \(\Cst\)\nb-algebra
  is simply the corner in \(\Cst_\Roe(X)\)
  generated by the projection onto~\(\Hils_{Y,R}\).
  Theorem~\ref{the:coarsely_dense} gives a unitary
  \(U\colon \Hils_Y \congto \Hils_{Y,R}\)
  such that \(U \Cst_\Roe(Y) U^* = \Cst_\Roe(Y_R)\).
  The unitary~\(U\)
  is built in the proof of Theorem~\ref{the:Roe_ample_unique}, and the
  construction there shows that it is controlled as an operator on
  \(\Hils_Y\oplus \Hils_X\).
  So it is a multiplier of~\(\Cst_\Roe(X)\),
  where it is no longer unitary but a partial isometry.  The operator
  \(U^* T U\)
  belongs to \(\Cst_\Roe(Y)\)
  because \(T\in\Cst_\Roe(Y_R)\).
  Since multipliers of \(\Cst_\Roe(X)\)
  are also multipliers of any ideal in~\(\Cst_\Roe(X)\),
  the operator \(T = U (U^* T U) U^*\)
  belongs to the ideal in~\(\Cst_\Roe(X)\)
  generated by \(\Cst_\Roe(Y) = P \Cst_\Roe(X) P\).
  Thus \(\Cst_\Roe(Y_R) \subseteq \Cst_\Roe(X) P \Cst_\Roe(X)\).
\end{proof}

\subsection{The coarse Mayer--Vietoris sequence}
\label{sec:coarse_MV_subsection}

\begin{proposition}[\cite{Higson-Roe-Yu:Coarse_Mayer-Vietoris}]
  \label{pro:omega-excisiv}
  Let~\(X\)
  be a proper metric space and let \(Y_1,Y_2\subseteq X\)
  be closed subspaces with \(Y_1 \cup Y_2 = X\).
  Let \(Z \defeq Y_1\cap Y_2\).
  Then
  \[
  \Cst_\Roe(Y_1\subseteq X) + \Cst_\Roe(Y_2\subseteq X) = \Cst_\Roe(X).
  \]
  We have
  \(\Cst_\Roe(Y_1\subseteq X) \cap \Cst_\Roe(Y_2\subseteq X) =
  \Cst_\Roe(Z\subseteq X)\)
  if and only if the following coarse transversality condition holds:
  for any \(R>0\)
  there is \(S(R)>0\)
  such that if \(x\in X\)
  satisfies \(d(x,Y_1)<R\)
  and \(d(x,Y_2)<R\),
  then \(d(x,Z)<S(R)\).
  The statements above hold both for real and complex Roe
  \(\Cst\)\nb-algebras.
\end{proposition}

\begin{proof}
  The proof of Theorem~\ref{the:subset_full_corner} identifies
  \(\Cst_\Roe(Y_j\subseteq X)\)
  for \(j=1,2\)
  with corners in \(\Cst_\Roe(X)\).
  The ideal generated by these corners is all of \(\Cst_\Roe(X)\)
  because \(Y_1 \cup Y_2 = X\).
  That is,
  \(\Cst_\Roe(Y_1\subseteq X) + \Cst_\Roe(Y_2\subseteq X) =
  \Cst_\Roe(X)\).

  An operator that is supported near~\(Z\)
  is also supported near~\(Y_1\) and near~\(Y_2\).  So
  \[
  \Cst_\Roe(Z\subseteq X)
  \subseteq \Cst_\Roe(Y_1\subseteq X)\cap \Cst_\Roe(Y_2\subseteq X),
  \]
  as these are closed ideals.  The other inclusion uses the coarse
  transversality assumption above, which is called
  ``\(\omega\)-excisiveness''
  in~\cite{Higson-Roe-Yu:Coarse_Mayer-Vietoris}.

  Let \(T\)
  and~\(U\)
  be locally compact operators that are \(R_T\)-
  and \(R_U\)\nb-controlled,
  respectively, and such that~\(T\)
  is supported within distance \(P_T>0\)
  of~\(Y_1\)
  and~\(U\)
  within distance \(P_U>0\)
  of~\(Y_2\).
  Let \(R \defeq R_T+R_U+P_T+P_U\).
  If \((x,y)\in\supp(T U)\),
  then there is \(z\in X\)
  with \((x,z)\in \supp(T)\) and \((z,y)\in \supp(U)\).  Then
  \begin{align*}
    d(x,Y_1)&<P_T\le R,\\
    d(x,Y_2)&\le d(x,z)+d(z,Y_2)< R_T+P_U\le R.
  \end{align*}
  Hence \(d(x,Y_1\cap Y_2)<S(R)\).
  A similar argument shows that \(d(y,Y_1\cap Y_2)<S(R)\).
  So \(T U\)
  is supported near~\(Y_1 \cap Y_2\).
  Thus
  \(\Cst_\Roe(Y_1\subseteq X)\cdot \Cst_\Roe(Y_2\subseteq X) \subseteq
  \Cst_\Roe(Z\subseteq X)\).
  This implies the claim because \(I\cap J=I\cdot J\)
  if \(I,J\)
  are closed ideals in a (real) \(\Cst\)\nb-algebra;
  the latter follows from the existence of approximate units.
\end{proof}

\begin{proposition}
  \label{pro:cstarmv}
  Let \(A\)
  be a real or complex \(\Cst\)\nb-algebra
  and let \(I,J\subset A\)
  be closed ideals with \(I+J=A\).
  Let \(\alpha_I\colon I\cap J\injto I\),
  \(\alpha_J\colon I\cap J\injto J\),
  \(\beta_I\colon I\injto A\),
  and \(\beta_J\colon J\injto A\)
  denote the inclusion maps and also the maps that they induce on
  \(\K\)\nb-theory.
  Then there is a long exact sequence in \textup{(}real or complex\textup{)}
  \(\K\)\nb-theory
  \[
  \dotsb \to \K_j(I\cap J)
  \xrightarrow{\bigl(\begin{smallmatrix} -\alpha_I\\\alpha_J \end{smallmatrix}\bigr)}
  \K_j(I)\oplus \K_j(J)
  \xrightarrow{\bigl(\begin{smallmatrix} \beta_I& \beta_J \end{smallmatrix}\bigr)}
  \K_j(A) \xrightarrow{\partial_\mathrm{MV}}
  \K_{j-1}(I\cap J) \to \dotsb.
  \]
  The boundary map~\(\partial_\mathrm{MV}\)
  is computed in~\eqref{eq:MV_boundary} below.
\end{proposition}

\begin{proof}
  There is a commuting diagram
  \[
  \begin{tikzcd}
    I\cap J \arrow[d, "\alpha_J"] \arrow[r, rightarrowtail, "\alpha_I"] &
    I \arrow[d, "\beta_I"] \arrow[r, twoheadrightarrow, "\pi_I"] &
    I/(I\cap J) \arrow[d, "\beta_{I*}", "\cong"'] \\
    J \arrow[r, rightarrowtail, "\beta_J"] &
    A \arrow[r, twoheadrightarrow, "\pi"] &
    A/J
  \end{tikzcd}
  \]
  whose rows are \(\Cst\)\nb-algebra
  extensions.  The map between the quotients induced by~\(\beta_I\)
  is an isomorphism because \(I/(I\cap J) \cong (I+J)/J = A/J\).
  The rows in the above diagram generate \(\K\)\nb-theory
  long exact sequences, which we view as exact chain complexes.  The
  vertical maps generate a chain map between them.  Its mapping cone
  is again exact.  So we get an exact sequence
  \begin{multline*}
    \dotsb \to \K_j(I\cap J) \oplus \K_{j+1}(A/J)
    \xrightarrow{\bigl(\begin{smallmatrix} -\alpha_I&0\\\alpha_J&\delta\end{smallmatrix}\bigr)}
    \K_j(I) \oplus \K_j(J)\\
    \xrightarrow{\bigl(\begin{smallmatrix} -\pi_I&0\\\beta_I&\beta_J\end{smallmatrix}\bigr)}
    \K_j(I/(I\cap J)) \oplus \K_j(A)
    \xrightarrow{\bigl(\begin{smallmatrix} -\delta&0\\\beta_{I*}&\pi\end{smallmatrix}\bigr)}
    \K_{j-1}(I\cap J) \oplus \K_j(A/J)
    \to \dotsb.
  \end{multline*}
  Since~\(\beta_{I*}\)
  is invertible, the boundary map restricts to an injective map on the
  summands \(\K_*(I/(I\cap J))\).
  So these summands and their images~\(B\)
  under the boundary map form an exact subcomplex.  Dividing it out
  gives another exact chain complex.  The
  direct summand \(\K_{j-1}(I\cap J)\)
  in \(\K_{j-1}(I\cap J) \oplus \K_j(A/J)\)
  is complementary to~\(B\),
  and the projection to \(\K_{j-1}(I\cap J)\)
  that kills~\(B\)
  maps \(x\in \K_j(A/J)\)
  to \(\delta(\beta_{I*}^{-1}(x))\)
  because the boundary map sends
  \(\beta_{I*}^{-1}(x) \in \K_j(I/(I\cap J))\)
  to \((-\delta(\beta_{I*}^{-1}(x)),x)\).
  Hence the quotient of the above complex by the exact subcomplex
  \(\K_*(I/(I\cap J)) \oplus B\) becomes the exact chain complex
  \begin{multline*}
    \dotsb \to \K_j(I\cap J)
    \xrightarrow{\bigl(\begin{smallmatrix} -\alpha_I\\\alpha_J\end{smallmatrix}\bigr)}
    \K_j(I) \oplus \K_j(J)\\
    \xrightarrow{\bigl(\begin{smallmatrix} \beta_I&\beta_J\end{smallmatrix}\bigr)}
    \K_j(A)
    \xrightarrow{\delta(\beta_{I*})^{-1}\pi}
    \K_{j-1}(I\cap J)
    \to \dotsb.
  \end{multline*}
  This is the desired long exact sequence.  We have also computed the
  boundary map:
  \begin{equation}
    \label{eq:MV_boundary}
    \partial_\mathrm{MV} = \delta\circ (\beta_{I*})^{-1}\circ \pi\colon
    \K_j(A) \xrightarrow{\pi} \K_j(A/J)
    \xrightarrow[\cong]{\beta_{I*}^{-1}} \K_j(I/(I\cap J))
    \xrightarrow{\delta} \K_{j-1}(I\cap J),
  \end{equation}
  where~\(\delta\)
  is the boundary map for the \(\Cst\)\nb-extension
  \(I\cap J \into I \prto I/(I\cap J)\).
\end{proof}

\begin{corollary}
  \label{cor:MV_coarse}
  Let~\(X\)
  be a proper metric space.  Let \(X= Y_1 \cup Y_2\)
  be a coarsely transverse decomposition as in
  Proposition~\textup{\ref{pro:omega-excisiv}} and let
  \(Z\defeq Y_1\cap Y_2\).  Then there is a long exact sequence
  \begin{multline*}
    \dotsb \to \K_j(\Cst_\Roe(Z))
    \xrightarrow{\bigl(\begin{smallmatrix} -\alpha_1\\\alpha_2 \end{smallmatrix}\bigr)}
    \K_j(\Cst_\Roe(Y_1))\oplus \K_j(\Cst_\Roe(Y_2))
    \\ \xrightarrow{\bigl(\begin{smallmatrix} \beta_1& \beta_2 \end{smallmatrix}\bigr)}
    \K_j(\Cst_\Roe(X)) \xrightarrow{\partial_\mathrm{MV}}
    \K_{j-1}(\Cst_\Roe(Z)) \to \dotsb.
  \end{multline*}
  Here \(\alpha_1,\alpha_2,\beta_1,\beta_2\)
  are the maps on \(\K\)\nb-theory
  induced by the \Star{}homomorphisms on Roe \(\Cst\)\nb-algebras
  induced by the inclusion maps \(Z\to Y_1\),
  \(Z\to Y_2\),
  \(Y_1\to X\)
  and \(Y_2\to X\),
  respectively.  The above holds both for real and complex Roe
  \(\Cst\)\nb-algebras.
\end{corollary}

\begin{proof}
  Let \(A\defeq \Cst_\Roe(X)\),
  \(I \defeq \Cst_\Roe(Y_1 \subseteq X)\)
  and \(J \defeq \Cst_\Roe(Y_2 \subseteq X)\).
  Then \(I+J=A\)
  and \(I\cap J = \Cst_\Roe(Z\subseteq X)\)
  by Proposition~\ref{pro:omega-excisiv}.  And the relative Roe
  \(\Cst\)\nb-algebras
  above are Morita equivalent to the absolute ones by
  Theorem~\ref{the:subset_full_corner}.  Thus
  \(\K_*(\Cst_\Roe(Y_j \subseteq X)) \cong \K_*(\Cst_\Roe(Y_j))\)
  for \(j=1,2\)
  and \(\K_*(\Cst_\Roe(Z \subseteq X)) \cong \K_*(\Cst_\Roe(Z))\).
  Plugging this into the Mayer--Vietoris sequence in
  Proposition~\ref{pro:cstarmv} gives the assertion.
\end{proof}

\subsection{Application to \texorpdfstring{$\Z^d$}{the plane}}
\label{sec:K-theory}

We apply the coarse Mayer--Vietoris sequence to the decomposition
\begin{equation}
  \label{eq:decompose_Zn}
  \Z^d = \Z^{d-1} \times \N \cup \Z^{d-1} \times (-\N)
\end{equation}
into two half-spaces, which intersect in \(\Z^{d-1} \times \{0\}\).
It is clearly coarsely transverse, so that
Corollary~\ref{cor:MV_coarse} applies to it.

\begin{proposition}[\cite{Higson-Roe-Yu:Coarse_Mayer-Vietoris}*{Proposition~1}]
  \label{vanish}
  The \(\K\)\nb-theory
  of the real and complex Roe \(\Cst\)\nb-algebras
  of \(X\times\N\) vanishes for any proper metric space~\(X\).
\end{proposition}

\begin{proof}
  We sketch the proof in~\cite{Higson-Roe-Yu:Coarse_Mayer-Vietoris}
  for~\(\K_0\)
  and then explain briefly why this argument also works for all other
  \(\K\)\nb-groups.
  We realise the Roe \(\Cst\)\nb-algebra
  on \(\ell^2(X\times\N,\Hils)\)
  for a separable Hilbert space~\(\Hils\),
  which may be real or complex.  The unilateral shift on
  \(\ell^2(\N)\)
  is a \(1\)\nb-controlled
  isometry.  It also defines a controlled isometry~\(S\)
  on \(\ell^2(X\times\N,\Hils)\).
  If \(n\in\N\),
  then the map \(T\mapsto S^n T (S^*)^n\)
  on \(\Cst_\Roe(X\times\N)\)
  is a \Star{}homomorphism, and it maps \(R\)\nb-controlled
  operators to \(R\)\nb-controlled
  operators for the same~\(R\).  Hence the map
  \[
  \varphi\colon \Bound(\ell^2(X\times\N,\Hils)) \to
  \Bound(\ell^2(X\times\N,\Hils^\infty)),\qquad
  T\mapsto \bigoplus_{n=0}^\infty S^n T (S^*)^n,
  \]
  maps \(R\)\nb-controlled
  operators on~\(\Hils\)
  to \(R\)\nb-controlled
  operators on \(\Hils^\infty \defeq \Hils\otimes \ell^2(\N)\).
  The matrix coefficients \((S^n T (S^*)^n)_{x,y}\)
  vanish for \(0\le x,y<n\).
  Therefore, \(\varphi(T)\)
  is locally compact if~\(T\)
  is locally compact.  So~\(\varphi\)
  restricts to a \Star{}homomorphism from the Roe \(\Cst\)\nb-algebra
  of~\(X\times\N\)
  realised on \(\ell^2(X\times\N,\Hils)\)
  to the isomorphic Roe \(\Cst\)\nb-algebra
  of~\(X\times\N\)
  realised on \(\ell^2(X\times\N,\Hils^\infty)\).
  We may identify these using a unitary operator
  \(\Hils \cong \Hils^\infty\).
  We have \(S\varphi(T)S^* = \bigoplus_{n=1}^\infty S^n T (S^*)^n\).
  So~\(\varphi\)
  is equal to the direct sum of the canonical inclusion~\(\iota\)
  induced by the embedding \(\Hils\to \Hils^\infty\),
  \(\xi\mapsto \xi\otimes \delta_0\),
  and the \Star{}homomorphism \(T\mapsto S\varphi(T) S^*\).

  In particular, if \(P\in\Cst_\Roe(X\times\N)\)
  is a projection, then \(\varphi(P)\)
  is another projection in \(\Cst_\Roe(X\times\N)\).
  And \(\varphi(P)\)
  is Murray--von Neumann equivalent to \(\iota(P) \oplus \varphi(P)\).
  Thus \(\iota(P)\)
  is stably equivalent to~\(0\).
  That is, the inclusion~\(\iota\)
  induces the zero map on~\(\K_0\).
  We may also identify \(\Hils^\infty \cong \Hils^2\)
  so that the inclusion~\(\iota\) becomes the corner embedding
  \[
  \Cst_\Roe(X\times\N) \to \Mat_2(\Cst_\Roe(X\times\N)),\qquad
  T\mapsto \begin{pmatrix} T&0\\0&0 \end{pmatrix}.
  \]
  This induces an isomorphism on~\(\K_0\).
  So the zero map is an isomorphism on \(\K_0(\Cst_\Roe(X\times\N))\),
  which forces this group to vanish.  The argument above works for any
  functor on the category of \(\Cst\)\nb-algebras
  and \Star{}homomorphisms that is matrix-stable because inner
  endomorphisms induced by isometries act by the identity on all such
  functors (see \cite{Cuntz-Meyer-Rosenberg}*{Proposition 3.16}).  In
  particular, the proof above works for all real and complex
  \(\K\)\nb-groups.
\end{proof}

The proof above breaks down for the uniform Roe \(\Cst\)\nb-algebra,
and indeed the result is wrong in that case.

\begin{corollary}
  \label{cor:K_Roe_Zd}
  Let \(d\in\N\).
  Then the boundary map in the coarse Mayer--Vietoris sequence for the
  decomposition~\eqref{eq:decompose_Zn} is an isomorphism.  So
  \(\K_{i+d}(\Cst_\Roe(\Z^d)_\mathbb{F}) \cong \K_i(\mathbb{F})\)
  for \(\mathbb{F} \in \{\R,\C\}\).
\end{corollary}

\begin{proof}
  Apply the Mayer--Vietoris long exact sequence of
  Corollary~\ref{cor:MV_coarse} to the coarsely transverse
  decomposition~\eqref{eq:decompose_Zn}.  Plug in that the
  \(\K\)\nb-theory
  vanishes for the two half-spaces (Proposition~\ref{vanish}).
  Hence the boundary map is an
  isomorphism.  Now an induction argument identifies
  \(\K_{i+d}(\Cst_\Roe(\Z^d)_\mathbb{F})\)
  with \(\K_i(\Cst_\Roe(\{0\})_\mathbb{F})\).
  Finally, \(\Cst_\Roe(\{0\})_\mathbb{F}\)
  is isomorphic to the \(\Cst\)\nb-algebra
  of compact operators on a real or complex Hilbert space.  This gives
  the statement because \(\K\)\nb-theory is \(\Cst\)\nb-stable.
\end{proof}

The well known \(K\)\nb-theory computations for \(\R\) and~\(\C\) give
\begin{align*}
  \K_i(\Cst_\Roe(\Z^d)_\C)&\cong
  \begin{cases} \Z &\text{if }i-d\equiv 0 \mod 2,\\
    0 &\text{if }i-d \equiv 1\mod 2.
  \end{cases}\\
  \K_i(\Cst_\Roe(\Z^d)_\R)&\cong
  \begin{cases}
    \Z &\text{if }i-d\equiv 0 \text{ or } 4 \mod 8,\\
    \Z/2&\text{if }i-d\equiv 1 \text{ or } 2 \mod 8,\\
    0 &\text{if } i-d\equiv 3\text{, }5\text{, }6\text{ or }7 \mod 8.
  \end{cases}
\end{align*}

In contrast, the \(\K\)\nb-theory
of the uniform \(\Cst\)\nb-Roe
algebra is far more complicated.  The \(\K_0\)\nb-group
of the complex uniform Roe \(\Cst\)\nb-algebra
of~\(\Z^d\)
is an uncountable Abelian group for all \(d>1\),
see \cite{Spakula:Thesis}*{Example II.3.4}.

When we consider Hamiltonians with symmetries, then we should tensor
the real Roe \(\Cst\)\nb-algebra
of~\(\Z^d\)
with a real or complex Clifford algebra.  This gives a
\(\Z/2\)\nb-graded
\(\Cst\)\nb-algebra.
Up to Morita equivalence, there are ten
different real or complex Clifford algebras.  So we get ten
different observable algebras in each dimension.  The resulting
real or complex \(\K\)\nb-groups
agree with those in Kitaev's periodic
table~\cite{Kitaev:Periodic_table}.  Hence the latter agrees with
the \(\K\)\nb-theory
of the Roe \(\Cst\)\nb-algebra.
We interpret it as saying that Kitaev's table gives only the
\emph{strong} topological phases.

The real and complex \(\K\)\nb-groups
of the point form a graded commutative, graded ring in a
natural way, and the \(\K\)\nb-theory
of any real or complex \(\Cst\)\nb-algebra
is a graded module over this ring.  The boundary map for an extension
of real or complex \(\Cst\)\nb-algebras
automatically preserves this module structure.  In the
complex case, the relevant ring is the ring of Laurent
polynomials \(\Z[\beta,\beta^{-1}]\)
in \(\beta\in \K_2(\C)\)
that describes Bott periodicity.  That a map on \(\K\)\nb-theory
is a \(\K_*(\C)\)-module
homomorphism only says that it is obtained by the maps on \(\K_0\)
and~\(\K_1\)
and Bott periodicity.  In other words, it is a homomorphism of
\(\Z/2\)\nb-graded groups.
In the real case, the relevant ring is more complicated, and so the
module structure contains more useful information.
One way to get the \(\K_*(\R)\)-module structure on \(\K_*(A)\)
for a real \(\Cst\)\nb-algebra~\(A\)
is to identify~\(\K_j(A)\)
with the bivariant Kasparov groups
\(\K_j(A) \cong \KK_0(\R,A \otimes \Cliff_j)\).
The exterior product in Kasparov theory provides both the graded
commutative ring structure on
\(\bigoplus_{j\in\Z} \KK_0(\R,\Cliff_j)\)
and the module structure on
\(\bigoplus_{j\in\Z} \KK_0(\R,A\otimes \Cliff_j)\).
These structures are compatible with Kasparov products, and the
boundary map in an extension may be written as such a Kasparov
product.

The isomorphism
\(\K_{*+d}(\Cst_\Roe(\Z^d)_\R) \cong \K_*(\R)\)
of \(\Z\)\nb-graded groups
in Corollary~\ref{cor:K_Roe_Zd} is a \(\K_*(\R)\)-module
isomorphism.  So \(\K_*(\Cst_\Roe(\Z^d)_\R)\)
is a free \(\K_*(\R)\)-module
of rank~\(1\),
shifted in degree by~\(d\).
And the boundary map~\(\partial_\mathrm{MV}\)
is a module isomorphism.  Thus it is determined by a single sign,
describing whether the ``standard'' generator of
\(\K_d(\Cst_\Roe(\Z^d)_\R)\)
goes to the ``standard'' generator of
\(\K_{d-1}(\Cst_\Roe(\Z^{d-1})_\R)\)
or its negative.  This sign is, in fact, a matter of convention: it
changes when we change the role of the left and right half-spaces in
the Mayer--Vietoris sequence.  So there is not much need to
``compute'' the boundary map for the Roe \(\Cst\)\nb-algebras
because the \(\K\)\nb-theory
groups in question are so small, even in the real case.

The boundary map~\(\partial_\mathrm{MV}\)
is the incarnation of the bulk--edge correspondence in our
Roe \(\Cst\)\nb-algebra
context.  It is shown by Kubota~\cite{Kubota:Controlled_bulk-edge}
that the boundary maps in the Toeplitz extension, which is used by
many authors to describe the bulk--edge correspondence, and the
coarse Mayer--Vietoris sequence are compatible.

\begin{proposition}
  \label{pro:low_dimensional_killed_Roe}
  Let \(\varphi\colon \Z^{d-1}\to \Z^d\)
  be an injective group homomorphism.  Then the induced map
  \(\varphi_*\colon \Cst_\Roe(\Z^{d-1}) \to \Cst_\Roe(\Z^d)\)
  induces the zero map in \(\K\)\nb-theory, both in the real and
  complex cases.
\end{proposition}

\begin{proof}
  Since~\(\varphi\)
  is an injective group homomorphism, it is a coarse equivalence
  from~\(\Z^{d-1}\)
  onto a subspace of~\(\Z^d\).
  This explains the definition of
  \(\varphi_*\colon \Cst_\Roe(\Z^{d-1}) \to \Cst_\Roe(\Z^d)\).
  There is \(x\in\Z^d\)
  so that the map \(\Z^{d-1}\times\Z \to \Z^d\),
  \((a,b) \mapsto \varphi(a)+b\cdot x\),
  is injective.  So the map
  \(\varphi_*\colon \Cst_\Roe(\Z^{d-1}) \to \Cst_\Roe(\Z^d)\)
  factors through \(\Cst_\Roe(\Z^{d-1}\times\N)\).
  Since the \(\K\)\nb-theory
  of \(\Cst_\Roe(\Z^{d-1}\times\N)\)
  vanishes by Proposition~\ref{vanish},
  the map~\(\varphi\)
  induces the zero map on \(\K\)\nb-theory.
\end{proof}

\section{Comparison with the periodic case}
\label{sec:weak_phases}

Let \(\mathbb{F} \in \{\R,\C\}\).
The observable algebra \(\Cst(\Z^d)_\mathbb{F}\)
or a matrix algebra over it describes periodic observables in the
limiting case of no disorder, in the tight-binding approximation.
This is contained in the corresponding Roe \(\Cst\)\nb-algebra
\(\Cst_\Roe(\Z^d)_\mathbb{F}\).
In this section, we recall how to compute the \(\K\)\nb-theory
of \(\Cst_\Roe(\Z^d)_\mathbb{F}\)
and we describe the map in \(\K\)\nb-theory
induced by the inclusion
\(\Cst(\Z^d)_\mathbb{F} \injto \Cst_\Roe(\Z^d)_\mathbb{F}\).
In particular, we show that this map is split surjective and that its
kernel is generated by those elements that come from the
\(\K\)\nb-theory
of \(\Cst(\Z^{d-1})_\mathbb{F}\)
for a coordinate embedding \(\Z^{d-1} \to \Z^d\).
So its kernel consists of those topological phases that are obtained
by stacking lower-dimensional topological insulators in a coordinate
direction.  These are called ``weak topological phases''
in~\cite{Fu-Kane-Mele:Insulators}.  In the end, we argue that stable
homotopy instead of homotopy is the physically reasonable equivalence
relation on Hamiltonians.

The following arguments are easier and more standard in the complex
case.  Hence we only discuss the real case.  It is convenient to
replace real \(\Cst\)\nb-algebras by ``real'' ones, that is, complex
\(\Cst\)\nb-algebras equipped with a real involution.  We first
recall some basic facts and definitions about ``real'' and real
\(\Cst\)\nb-algebras and then describe the relevant ``real''
\(d\)\nb-torus.

A real \(\Cst\)\nb-algebra~\(A\)
corresponds to the ``real'' \(\Cst\)\nb-algebra
\(A\otimes_\R \C\)
with the real involution
\(\conj{a\otimes z} \defeq a\otimes \conj{z}\).
A ``real'' \(\Cst\)\nb-algebra~\(A\)
corresponds to the real \(\Cst\)\nb-algebra
\[
A_\R \defeq \setgiven{a\in A}{\conj{a}=a}.
\]
A ``real'' locally compact space~\(X\)
is a locally compact space with an involutive homeomorphism
\(X\to X\),
\(x\mapsto \conj{x}\).
Then we turn \(\Cont_0(X)\)
into a ``real'' \(\Cst\)\nb-algebra
using the real involution \(\conj{f}(x) \defeq \conj{f(\conj{x})}\)
for all \(x\in X\), \(f\in\Cont_0(X)\).  So
\[
\Cont_0(X)_\R = \setgiven*{f\in\Cont_0(X)}
{f(\conj{x}) = \conj{f(x)}\text{ for all }x\in X}.
\]
For a ``real'' \(\Cst\)\nb-algebra~\(A\), we define
\[
\KR_*(A) \defeq \K_*(A_\R).
\]
For a ``real'' locally compact space~\(X\), we let
\[
\KR^*(X)
\defeq \KR_{-*}(\Cont_0(X))
= \K_{-*}(\Cont_0(X)_\R).
\]
Note the grading convention here, which is analogous to the numbering
convention when a chain complex is treated as a cochain complex.

From now on, \(\Cst(\Z^d)\)
denotes the ``real'' \(\Cst\)\nb-algebra
that corresponds to the real \(\Cst\)\nb-algebra
\(\Cst(\Z^d)_\R\).
That is, the real involution acts on \(f\colon \Z^d\to\C\)
by pointwise complex conjugation.  We give the \(d\)\nb-torus
\(\T^d\subseteq \C^d\)
the real involution by complex conjugation.  So \(\Cont(\T^d)_\R\)
is the closed \(\R\)\nb-linear
span of the functions \(z^k \defeq z_1^{k_1} \dotsm z_d^{k_d}\)
on~\(\T^d\)
for \(k_1,\dotsc,k_d\in\Z\).
This is the unique real structure on~\(\T^d\)
for which the Fourier isomorphism \(\Cst(\Z^d) \cong \Cont(\T^d)\)
is an isomorphism of ``real'' \(\Cst\)\nb-algebras.  Thus
\[
\K_*(\Cst(\Z^d)_\R) \cong \KR_*(\Cont(\T^d)) = \KR^{-*}(\T^d).
\]
We shall also use the ``real'' manifolds~\(\R^{p,q}\)
for \(p,q\in\N\);
this is~\(\R^{p+q}\)
with the real involution \(\conj{(x,y)} \defeq (x,-y)\)
for \(x\in\R^p\),
\(y\in\R^q\).
We may also realise this as
\(\R^p \times (\ima \R)^q \subseteq \C^{p+q}\)
with complex conjugation as real involution.

\begin{proposition}
  \label{pro:Cst_Zd_in_KK}
  The ``real'' \(\Cst\)\nb-algebra
  \(\Cont(\T)\)
  is \(\KK\)\nb-equivalent
  to \(\C\oplus \Cont_0(\R^{0,1})\).
  And \(\Cont(\T^d)\)
  is \(\KK\)\nb-equivalent
  to a direct sum of copies of \(\Cont_0(\R^{0,j})\)
  for \(j=0,\dotsc,d\),
  where the summand \(\Cont_0(\R^{0,j})\)
  appears \(\binom{d}{j}\) times.
  The \(\K\)\nb-theory
  of \(\Cst(\Z^d)_\mathbb{F}\)
  is a free \(\K_*(\mathbb{F})\)-module
  of rank~\(2^d\),
  with \(\binom{d}{j}\) generators of degree \(-j \bmod 8\).
\end{proposition}

\begin{proof}
  The points \(\pm 1\in\T\)
  are real, that is, fixed by the real involution.  The complement
  \(\T\setminus\{1\}\)
  is diffeomorphic as a ``real'' manifold to~\(\R^{0,1}\),
  say, by stereographic projection at~\(1\).
  Hence we get an extension of ``real'' \(\Cst\)\nb-algebras
  \[
  \Cont_0(\R^{0,1}) \into \Cont(\T)
  \prto \C,
  \]
  where the quotient map is evaluation at~\(1\).
  This extension splits by embedding~\(\C\)
  as constant functions in \(\Cont(\T)\).
  Since Kasparov theory is split-exact, also for ``real''
  \(\Cst\)\nb-algebras,
  \(\Cont(\T)\)
  is \(\KK\)\nb-equivalent to \(\Cont_0(\R^{0,1}) \oplus \C\).

  We may get \(\Cst(\Z^d)\)
  by tensoring \(d\)~copies
  of \(\Cst(\Z)\).
  The tensor product of \(\Cst\)\nb-algebras
  descends to a bifunctor in \(\KK\)\nb-theory,
  also in the ``real'' case.  So \(\Cst(\Z^d)\)
  is \(\KK\)\nb-equivalent
  to the \(d\)\nb-fold
  tensor power of \(\C \oplus \Cont_0(\R^{0,1})\).
  The tensor product of \(\Cst\)\nb-algebras
  is additive in each variable, and
  \(\Cont_0(\R^{p,q})\otimes \Cont_0(\R^{r,s}) \cong
  \Cont_0(\R^{p+r,q+s})\).
  A variant of the binomial formula now gives.
  \[
  \Cst(\Z^d)
  \sim_\KK (\R \oplus \Cont_0(\R^{0,1}))^{\otimes d}
  \cong \bigoplus_{j=0}^d \binom{d}{j} \Cont_0(\R^{0,j}).
  \]
  A Bott periodicity theorem by Kasparov shows that
  \(\Cont_0(\R^{p,q})\)
  is \(\KK\)\nb-equivalent
  to~\(\Cont_0(\R^{0,0})\)
  with a dimension shift of \(p-q\),
  see \cite{Kasparov:Operator_K}*{Theorem~7}.  This implies the claim
  about \(\K\)\nb-theory.
\end{proof}

\begin{proposition}
  \label{pro:low_dimensional_killed}
  Let \(\varphi\colon \Z^{d-1}\to \Z^d\)
  be an injective group homomorphism.  It induces an injective
  \Star{}homomorphism
  \(\varphi_*\colon \Cst(\Z^{d-1})_\mathbb{F} \to
  \Cst(\Z^d)_\mathbb{F}\)
  and a grading-preserving \(\K_*(\mathbb{F})\)-module
  homomorphism
  \(\K_*(\varphi_*)\colon \K_*(\Cst(\Z^{d-1})_\mathbb{F}) \to
  \K_*(\Cst(\Z^d)_\mathbb{F})\).
  The map
  \(\K_*(\Cst(\Z^d)_\mathbb{F}) \to \K_*(\Cst_\Roe(\Z^d)_\mathbb{F})\)
  vanishes on the image of~\(\K_*(\varphi_*)\).
\end{proposition}

\begin{proof}
  Proposition~\ref{pro:low_dimensional_killed_Roe} shows
  that~\(\varphi\)
  induces the zero map on the \(\K\)\nb-theory
  of the Roe \(\Cst\)\nb-algebra.
  The canonical map \(\Cred(G) \injto \Cst_\Roe(G)\)
  for a group~\(G\)
  is a natural transformation with respect to injective group
  homomorphisms.  So there is a commuting square
  \[
  \begin{tikzcd}
    \Cst(\Z^{d-1})_\mathbb{F}
    \arrow[r, hookrightarrow] \arrow[d, hookrightarrow, "\varphi_*"]&
    \Cst_\Roe(\Z^{d-1})_\mathbb{F} \arrow[d, hookrightarrow, "\varphi_*"]\\
    \Cst(\Z^d)_\mathbb{F}
    \arrow[r, hookrightarrow]&
    \Cst_\Roe(\Z^d)_\mathbb{F}
  \end{tikzcd}
  \]
  This implies the statement.
\end{proof}

The coordinate embeddings
\[
\iota_k\colon \Z^{d-1}\to \Z^d,\qquad
(x_1,\dotsc,x_{d-1}) \mapsto
(x_1,\dotsc,x_{k-1},0,x_k,\dotsc,x_{d-1}),
\]
are injective group homomorphisms and induce injective
\Star{}homomorphisms
\[
\iota_k\colon \Cst(\Z^{d-1}) \to \Cst(\Z^d).
\]
The Fourier transform maps \(\iota_k(\Cst(\Z^{d-1}))\)
onto the \(\Cst\)\nb-subalgebra
of \(\Cont(\T^d)\)
consisting of all functions that are constant equal to~\(1\)
in the \(k\)th
coordinate direction.  We have seen that
\(\K_*(\Cst(\Z^d)_\mathbb{F})\)
is a free \(\K_*(\mathbb{F})\)-module
of rank~\(2^d\)
(with generators in different degrees).  Now compute
\(\K_*(\Cst(\Z^d)_\mathbb{F})\)
as in Proposition~\ref{pro:Cst_Zd_in_KK}.  The inclusion of functions
that are constant in the \(k\)th
direction corresponds in \(\K\)\nb-theory
to the inclusion of those~\(2^{d-1}\)
of the~\(2^d\)
free \(\K_*(\R)\)-module
summands in \(\K_*(\Cst(\Z^d)_\mathbb{F})\)
where we take the summand~\(\R\)
in the \(k\)th
factor.  The summands in the image of~\(\K_*(\iota_k)\)
correspond to topological insulators that are built by stacking
copies of a \(d-1\)-dimensional
insulator in the \(k\)th
direction.  Such topological insulators are considered weak by
Fu--Kane--Mele~\cite{Fu-Kane-Mele:Insulators}.  So the map
\(\K_*(\Cst(\Z^d)_\mathbb{F}) \to \K_*(\Cst_\Roe(\Z^d)_\mathbb{F})\)
kills the \(\K\)\nb-theory classes of weak topological insulators.

If~\(k\)
varies, then all but one of the~\(2^d\)
summands \(\K_{*-j}(\mathbb{F})\)
in \(\K_*(\Cst(\Z^d)_\mathbb{F})\)
are in the image of \(\K_*(\iota_k)\)
for some \(k\in\{1,\dotsc,d\}\).
All these summands are mapped to~\(0\)
in \(\K_*(\Cst_\Roe(\Z^d)_\mathbb{F})\)
by Proposition~\ref{pro:low_dimensional_killed}.  The remaining
summand is the \(\K\)\nb-theory
of the ideal \(\Cont_0(\R^{0,d}) \idealin \Cont(\T^d)\).
Here we identify~\(\R^{0,d}\)
with an open subset of~\(\T^d\)
using the stereographic projection in each variable, compare the proof
of Proposition~\ref{pro:Cst_Zd_in_KK}.  Its ``real'' or complex
\(\K\)\nb-theory
is identified with \(\K_{*-d}(\R)\)
or \(\K_{*-d}(\C)\)
by Bott periodicity.  Kasparov proves Bott periodicity isomorphisms
\(\KR_*(\Cont_0(\R^{p,q})) \cong \K_{*+p-q}(\R)\)
using a canonical generator~\(\alpha_{p,q}\)
for the \(\K\)\nb-homology group
\[
\KK_{q-p}^\R(\Cont_0(\R^{p,q}),\C)
\cong \KK_0^\R(\Cont_0(\R^{p,q}, \Cliff_{p,q}),\C);
\]
here~\(\Cliff_{p,q}\)
is the Clifford algebra with \(p+q\)~anti-commuting,
odd, self-adjoint generators \(\gamma_1,\dotsc,\gamma_{p+q}\)
with \(\conj{\gamma_i} = \gamma_i\)
for \(1\le i\le p\)
and \(\conj{\gamma_i} = -\gamma_i\)
for \(p+1\le i \le p+q\).
And we write~\(\KK^\R\)
to highlight that the entries are treated as ``real''
\(\Cst\)\nb-algebras.

The Bott generators~\(\alpha_{p,0}\)
are generalised by Kasparov in \cite{Kasparov:Novikov}*{Definition
  and Lemma 4.2} to build a ``fundamental class''
\[
\alpha_X \in \KK^\R_0(\Cont_0(X,\Cliff X),\C)
\]
for any complete Riemannian manifold~\(X\)
(without boundary).  Here~\(\Cliff X\)
is the bundle of ``real'' \(\Cst\)\nb-algebras
over~\(X\)
whose fibre at \(x\in X\)
is the \(\Z/2\)\nb-graded
``real'' Clifford algebra of the cotangent space~\(T_x^* X\)
for the positive definite quadratic form induced by the Riemannian
metric, and \(\Cont_0(X,\Cliff X)\)
means the \(\Z/2\)\nb-graded
``real'' \(\Cst\)\nb-algebra
of \(\Cont_0\)\nb-sections
of this Clifford algebra bundle.  We now adapt Kasparov's fundamental
class to the case where~\(X\)
is a ``real'' complete Riemannian manifold, in such a way that the
fundamental class for~\(\R^{p,q}\)
is the generator~\(\alpha_{p,q}\)
of Bott periodicity from~\cite{Kasparov:Operator_K}.  The only changes
are in the real structure.  In particular, all the analysis needed to
produce cycles for Kasparov theory is already done
in~\cite{Kasparov:Novikov}.

Recall that the real involution on \(\Cont_0(X)\)
is defined by \(\conj{f}(x) \defeq \conj{f(\conj{x})}\)
for \(f\in\Cont_0(X)\).
There is a unique conjugate-linear involution on the space of
complex \(1\)\nb-forms on~\(X\)
such that \(\conj{\diff f} = \diff{\conj{f}}\)
for all smooth \(f\in\Cont_0(X)\).
There is a unique conjugate-linear involution on
\(\Cont_0(X,\Cliff X)\) with
\[
\conj{\omega_1 \dotsm \omega_m}
= \conj{\omega_1} \dotsm \conj{\omega_m}
\]
for all sections \(\omega_1,\dotsc,\omega_m\)
of \(T^* X \otimes\C\).
This involution is also compatible with the multiplication and the
\(\Z/2\)\nb-grading.  So it turns \(\Cont_0(X,\Cliff X)\)
into a \(\Z/2\)\nb-graded ``real'' \(\Cst\)\nb-algebra.

Let~\(L^2(\Lambda^*(X))\)
be the Hilbert space of square-integrable complex differential forms
on~\(X\).
This is the underlying Hilbert space of Kasparov's fundamental
class.  It is \(\Z/2\)\nb-graded
so that sections of~\(\Lambda^{2\ell}(X)\)
are even and sections of~\(\Lambda^{2\ell+1}(X)\)
are odd.  There is a unique conjugate-linear, isometric involution
on \(L^2(\Lambda^*(X))\) with
\[
\conj{\omega_1 \wedge \dotsb \wedge \omega_\ell} =
\conj{\omega_1} \wedge \dotsb \wedge \conj{\omega_\ell}
\]
for all complex \(1\)\nb-forms \(\omega_1,\dotsc,\omega_\ell\).
It commutes with the \(\Z/2\)\nb-grading,
so that \(L^2(\Lambda^*(X))\)
becomes a \(\Z/2\)\nb-graded
``real'' Hilbert space.

Given a complex \(1\)\nb-form~\(\omega\)
and a differential form~\(\eta\),
let \(\lambda_\omega(\eta) \defeq \omega \wedge \eta\).
These operators satisfy the relations
\begin{equation}
  \label{eq:lambda_and_star}
  \lambda_\omega \lambda_\eta + \lambda_\eta \lambda_\omega=0,\qquad
  \lambda_\omega^* \lambda_\eta + \lambda_\eta \lambda_\omega^*
  = \braket{\omega}{\eta}
\end{equation}
for all complex \(1\)\nb-forms \(\omega,\eta\),
where \(\braket{\omega}{\eta}\in\Cont_0(X)\)
denotes the pointwise inner product, which acts on
\(L^2(\Lambda^*(X))\)
by pointwise multiplication.
The representation of \(\Cont_0(X,\Cliff X)\)
on \(L^2(\Lambda^*(X))\)
is defined by letting a complex \(1\)\nb-form~\(\omega\),
viewed as an element of \(\Cont_0(X,\Cliff X)\),
act by \(\lambda_\omega + \lambda_{\omega^*}^*\).
Here~\(\omega^*\)
is the adjoint of~\(\omega\)
in the \(\Cst\)\nb-algebra \(\Cont_0(X,\Cliff X)\),
that is,
\(\omega^*(x) = \omega(x)^*\)
for all \(x\in X\),
where \(\omega(x)^* \in T^*_x X \otimes \C\)
is the pointwise complex conjugation in the second tensor
factor~\(\C\).
This defines a \Star{}representation of \(\Cont_0(X,\Cliff X)\)
by~\eqref{eq:lambda_and_star}.  It is grading-preserving and real as
well.

Let~\(d\) be the de Rham differential, defined on smooth
sections of~\(\Lambda^*(X)\)
with compact support, and let~\(d^*\)
be its adjoint.  The unbounded operator \(\mathcal{D} \defeq d+d^*\)
is essentially self-adjoint because~\(X\)
is complete.  So
\[
F \defeq (1+\mathcal{D}^2)^{-1/2} \mathcal{D}
\]
is a well defined self-adjoint operator.  The operator~\(d\) is odd
and real.  This is inherited by~\(\mathcal{D}\)
and~\(F\).
Kasparov shows that \((1-F^2)\cdot a\)
and \([F,a]\)
are compact for all \(a\in \Cont_0(X,\Cliff X)\).
Thus \(\alpha_X \defeq (L^2(\Lambda^*(X)),F)\)
is a cycle for the ``real'' Kasparov group
\(\KK^\R_0(\Cont_0(X,\Cliff X),\C)\).
We call this the \emph{fundamental class} of the ``real''
manifold~\(X\).
(Kasparov calls it ``Dirac element'' instead.)

In particular, the fundamental class of the ``real''
manifold~\(\R^{p,q}\)
becomes the Bott periodicity generator~\(\alpha_{p,q}\)
from~\cite{Kasparov:Operator_K} when we trivialise the Clifford
algebra bundle on~\(\R^{p,q}\)
in the obvious way.  So
\(\alpha_{\R^{p,q}}\in \KK^\R_0(\Cont_0(\R^{p,q}) \otimes
\Cliff_{p,q},\C)\)
is invertible.

We give~\(\T^d\)
the \(\T^d\)\nb-invariant
Riemannian metric to build its fundamental class.
The torus~\(\T^d\)
is parallelisable as a ``real'' manifold: its tangent bundle is
isomorphic to \(\T^d \times \R^{0,d}\).
This induces an isomorphism
\(\Cont(\T^d, \Cliff \T^d) \cong \Cont(\T^d) \otimes \Cliff_{0,d}\).
So the fundamental class~\(\alpha_{\T^d}\)
also gives an element in \(\KK^\R_d(\Cont(\T^d),\C)\).

Let~\(\Hils[L]\)
be a separable ``real'' Hilbert space and build \(\Cst_\Roe(\Z^d)\)
on the ``real'' Hilbert space \(\ell^2(\Z^d,\Hils[L])\).
There is an obvious embedding
\(\Cst(\Z^d) \otimes \Comp(\Hils[L]) \subseteq \Cst_\Roe(\Z^d)\).  Let
\[
\alpha_{\T^d}^{\Hils[L]} \in
\KK^\R_0(\Cont(\T^d) \otimes \Cliff_{0,d} \otimes \Comp(\Hils[L]), \C)
\cong
\KK^\R_0(\Cst(\Z^d) \otimes \Cliff_{0,d} \otimes \Comp(\Hils[L]), \C)
\]
be the exterior product of the fundamental class~\(\alpha_{\T^d}\)
and the Morita equivalence \(\Comp(\Hils[L]) \sim \C\).
This is the Kasparov cycle with underlying \(\Z/2\)\nb-graded
``real'' Hilbert space
\(L^2(\T^d, \Lambda^*(\C^d)) \otimes \Hils[L]\)
with the operator \(\tilde{F} \defeq F\otimes 1\)
with~\(F\)
as above for the manifold \(X=\T^d\).
So~\(\tilde{F}\)
is an odd, self-adjoint, real bounded operator with
\begin{equation}
  \label{eq:Kasparov_cycle}
  [\tilde{F},T],(1-\tilde{F}^2)\cdot T\in
  \Comp(\ell^2(\Z) \otimes \Lambda^*(\C^d) \otimes \Hils[L])
\end{equation}
for all
\(T\in\Cst(\Z^d) \otimes \Cliff_{0,d} \otimes \Comp(\Hils[L])\)
(the commutator is the graded one).

\begin{theorem}
  \label{the:fundamental_class}
  Equation~\eqref{eq:Kasparov_cycle} still holds for
  \(T\in\Cst_\Roe(\Z^d) \otimes \Cliff_{0,d}\).
  This gives
  \[
  \alpha'_{\T^d}
  \defeq [(L^2(\T^d, \Lambda^*(\C^d)) \otimes \Hils[L], \tilde{F})]
  \in \KK^\R_0(\Cst_\Roe(\Z^d) \otimes \Cliff_{0,d}, \C).
  \]
  The following diagram in~\(\KK^\R\) commutes:
  \[
  \begin{tikzcd}[column sep=2.3em]
    \Cont_0(\R^{0,d}, \Cliff_{0,d}) \arrow[r, "\textup{incl.}"]
    \arrow[dr, "\alpha_{\R^{0,d}}", "\cong"'] &
    \Cont(\T^d, \Cliff_{0,d}) \arrow[r, "\textup{Fourier}"]
    \arrow[d, "\alpha_{\T^d}"] &
    \Cst(\Z^d) \otimes \Cliff_{0,d} \arrow[r, "\textup{incl.}"] &
    \Cst_\Roe(\Z^d) \otimes \Cliff_{0,d} \arrow[dll, "\alpha'_{\T^d}"]
    \\ &\C
  \end{tikzcd}
  \]
\end{theorem}

\begin{corollary}
  \label{cor:K_periodic_split-injective}
  The inclusion \(\Cont_0(\R^{0,d}) \to \Cst_\Roe(\Z^d)\)
  induces a split injective map
  \(\KR_*(\Cont_0(\R^{0,d})) \to \KR_*(\Cst_\Roe(\Z^d))\).
  The map
  \(\K_{*+d}(\Cst_\Roe(\Z^d)_\mathbb{F}) \to \K_*(\mathbb{F})\)
  induced by~\(\alpha'_{\T^d}\)
  is an isomorphism.  Analogous statements hold in complex
  \(\K\)\nb-theory.
\end{corollary}

\begin{proof}[Proof of the corollary]
  Both \(\KR_*(\Cont_0(\R^{0,d}))\)
  and \(\KR_*(\Cst_\Roe(\Z^d))\)
  are isomorphic to free \(\K_*(\R)\)-modules
  with a generator in degree~\(-d\).
  The Bott periodicity generator~\(\alpha_{\R^{0,d}}\)
  maps the generator of \(\KR_{-d}(\Cont_0(\R^{0,d}))\)
  onto a generator of \(\K_0(\R)\).
  The commuting diagram in Theorem~\ref{the:fundamental_class} shows
  that its image in \(\KR_{-d}(\Cst_\Roe(\Z^d))\)
  must be a generator as well.  So~\(\alpha'_{\T^d}\)
  acts by multiplication with~\(\pm1\)
  on a generator.  Since~\(\alpha'_{\T^d}\)
  is a \(\K\)\nb-homology
  class, the map on \(\K\)\nb-theory
  that it induces is a \(\K_*(\R)\)-module
  homomorphism.  Hence it is multiplication by~\(\pm1\)
  everywhere once this happens on a generator.  So the map
  on~\(\KR_*\)
  induced by~\(\alpha'_{\T^d}\)
  is invertible.  The same proof works for complex \(\K\)\nb-theory.
\end{proof}

We have already shown that all but one of the free
\(\K_*(\mathbb{F})\)-module
summands in \(\K_*(\Cst(\Z^d)_\mathbb{F})\)
are killed by the map to \(\K_*(\Cst_\Roe(\Z^d)_\mathbb{F})\).
When we combine this with the above corollary, it follows that the
kernel of the map from \(\K_*(\Cst(\Z^d)_\mathbb{F})\)
to \(\K_*(\Cst_\Roe(\Z^d)_\mathbb{F})\)
is exactly the sum of the images of \(\K_*(\iota_k)\)
for \(k=1,\dotsc,d\),
that is, the subgroup generated by the \(\K\)\nb-theory
classes of weak topological insulators.

We still have to prove Theorem~\ref{the:fundamental_class}.  The
left triangle in the diagram in Theorem~\ref{the:fundamental_class}
commutes because of the following general fact:

\begin{proposition}
  \label{pro:fundamental_class_open_restrict}
  Let \(U\subseteq X\)
  be an open subset of a ``real'' manifold~\(X\)
  that is invariant under the real involution.
  Give \(U\)
  and~\(X\)
  some complete Riemannian metrics.
  The Kasparov product of the ideal inclusion
  \(j\colon \Cont_0(U,\Cliff U) \injto \Cont_0(X,\Cliff X)\)
  and the fundamental class
  \(\alpha_X\in \KK_0^\R(\Cont_0(X,\Cliff X), \C)\)
  is the fundamental class
  \(\alpha_U\in \KK_0^\R(\Cont_0(U,\Cliff U), \C)\).
  In particular, the fundamental class does not depend on the choice
  of the Riemannian metric.
\end{proposition}

\begin{proof}
  Both Kasparov cycles \(j^*(\alpha_X)\)
  and~\(\alpha_U\)
  live on Hilbert spaces of \(L^2\)\nb-differential
  forms on~\(U\).
  Here square-integrability is with respect to different metrics.
  The resulting Hilbert spaces are isomorphic by pointwise
  application of a suitable strictly positive, smooth function
  \(X\to \Bound(\Lambda^* X)\).
  This isomorphism also respects the \(\Z/2\)\nb-grading
  and the real involution.  The operator~\(\mathcal{D}\)
  used to construct~\(\alpha_X\)
  is a first-order differential operator.  Hence \(F \defeq
  (1+\mathcal{D}^2)^{-1/2} \mathcal{D}\)
  is an order-zero pseudodifferential operator, and it has the same
  symbol as~\(\mathcal{D}\).
  This is the function
  \[
  S^* X \to \Bound(\Lambda^*(X)),\qquad
  (x,\xi) \mapsto \lambda_\xi + \lambda_\xi^*.
  \]
  The symbol of the operator~\(F\)
  of~\(\alpha_U\)
  is given by the same formula, except that the adjoint is for
  another Riemannian metric.  So the unitary
  between the spaces of \(L^2\)\nb-forms
  will also identify these symbols.  The class of the Kasparov
  cycle defined by an order-zero pseudodifferential operator~\(F\)
  depends only on the symbol of~\(F\).
  So \(j^*(\alpha_X)\)
  and~\(\alpha_U\)
  have the same class in \(\KK_0^\R(\Cont_0(U,\Cliff U), \C)\).
  The last statement is the case \(U=X\)
  of the proposition.
\end{proof}

Now we build the Kasparov cycle
\(\alpha'_{\T^d}\in \KK_0^\R(\Cst_\Roe(\Z^d)\otimes \Cliff_{0,d},\C)\).
Let \(z_j\colon \T^d\to\C\)
be the \(j\)th coordinate function and let
\[
z^k \defeq z_1^{k_1} \dotsm z_d^{k_d}
\qquad\text{for }
k=(k_1,\dotsc,k_d)\in\Z^d.
\]
The real involution on~\(\T^d\)
is defined so that these are real elements of \(\Cont(\T^d)\).
Hence the \(1\)\nb-forms \(z_j^{-1} \diff z_j\)
for \(j=1,\dotsc,d\)
are real.  They form a basis of the space of \(1\)\nb-forms
as a \(\Cont(\T^d)\)-module.
The differential forms
\[
z^k\cdot (z_{i_1} \dotsm z_{i_\ell})^{-1}
\diff z_{i_1} \wedge \dotsb \wedge \diff z_{i_\ell}
\]
for \(k\in\Z^d\)
and \(1\le i_1 < i_2 < \dotsb < i_\ell\le d\)
form a real, orthonormal basis of the Hilbert space
\(L^2(\Lambda^*(\T^d))\).
Hence there is a unitary operator
\begin{multline*}
  U\colon
  \ell^2(\Z^d) \otimes \Lambda^*(\C^d) \congto
  L^2(\Lambda^*(\T^d)),\\
  \delta_k \otimes e_{i_1} \wedge \dotsb \wedge e_{i_\ell} \mapsto
  z^k\cdot (z_{i_1} \dotsm z_{i_\ell})^{-1}
  \diff z_{i_1} \wedge \dotsb \wedge \diff z_{i_\ell}.
\end{multline*}
This unitary is grading-preserving and real for the
\(\Z/2\)\nb-grading and real structure on
\(\ell^2(\Z^d) \otimes \Lambda^\ell(\C^d)\)
where the standard basis vector
\(\delta_k \otimes e_{i_1} \wedge \dotsb \wedge e_{i_\ell}\)
is real and is even or odd depending on the parity of~\(\ell\).

The above trivialisation of the cotangent bundle of~\(\T^d\)
gives the isomorphism
\[
\Cont(\T^d) \otimes \Cliff_{0,d} \congto \Cont(\T^d,\Cliff \T^d),
\qquad
\gamma_j\mapsto \ima z_j^{-1}\,\diff z_j;
\]
recall that \(\gamma_1,\dotsc,\gamma_d\)
are the odd, self-adjoint, anti-commuting unitaries that
generate~\(\Cliff_{0,d}\).
The action of \(\Cont(\T^d,\Cliff \T^d)\)
on \(L^2(\Lambda^* \T^d)\)
now translates to an action of
\(\Cont(\T^d) \otimes \Cliff_{0,d}\)
on \(\ell^2(\Z^d) \otimes \Lambda^*(\C^d)
\cong \ell^2(\Z^d, \Lambda^*(\C^d))\).
Namely, the scalar-valued function \(z^k\in\Cont(\T^d)\)
acts by the shift \((\tau_k f)(n) \defeq f(n-k)\)
for all \(k,n\in\Z^d\),
\(f\in \ell^2(\Z^d,\Lambda^*(\C^d))\).
And the Clifford generator \(\gamma_j\in\Cliff_{0,d}\)
acts by
\[
(\gamma_j f)(n)
= \ima \lambda_{e_j}\bigl(f(n)\bigr)
- \ima \lambda_{e_j}^*\bigl(f(n)\bigr).
\]
The unitary~\(U^*\)
maps the domain of~\(d\)
to the space of rapidly decreasing functions
\(\Z^d\to\Lambda^*(\C^d)\),
where~\(U^* d U\)
acts by pointwise application of the function
\[
A\colon \Z^d\to \Bound(\Lambda^*(\C^d)),\qquad
n \mapsto \lambda_n
= \sum_{j=1}^d n_j\cdot \lambda_{e_j},
\]
because
\[
d\left(z^k\cdot \frac{\diff z_{i_1}}{z_{i_1}} \wedge \dotsc
\wedge \frac{\diff z_{i_\ell}}{z_{i_\ell}}\right)
= \sum_{j=1}^d k_j z^k \cdot \frac{\diff z_j}{z_j}
\wedge \frac{\diff z_{i_1}}{z_{i_1}} \wedge \dotsc
\wedge \frac{\diff z_{i_\ell}}{z_{i_\ell}}.
\]
So~\(U^* \mathcal{D} U\)
acts by pointwise application of the matrix-valued function \(A+A^*\)
on the space of rapidly decreasing functions
\(\Z^d\to\Lambda^*(\C^d)\).
We compute
\[
(A+A^*)^2(n)
= \lambda_n \lambda_n^* + \lambda_n^* \lambda_n
= \norm{n}^2.
\]
So~\(U^* F U\)
acts by pointwise application of the matrix-valued function
\[
\hat\alpha_{\Z^d}\colon
\Z^d \to \Bound(\Lambda^*(\C^d)),\qquad
n \mapsto (1+\norm{n}^2)^{-1/2} (\lambda_n + \lambda_n^*).
\]

Next we take the exterior product with the Morita equivalence
between \(\Comp(\Hils[L])\)
and~\(\C\).
This simply gives the Hilbert space
\(\ell^2(\Z^d,\Lambda^*(\C^d)) \otimes \Hils[L]\)
with the induced \(\Z/2\)\nb-grading
and ``real'' structure, the exterior tensor product representation
of \(\Cont(\T^d)\otimes\Cliff_{0,d}\otimes \Comp(\Hils[L])\),
and with the operator \(F\otimes 1_{\Hils[L]}\).
This is a Kasparov cycle for
\(\KK^\R_0(\Cst(\Z^d) \otimes \Cliff_{0,d} \otimes \Comp(\Hils[L]),
\C)\).
In particular, the operator \(\tilde{F} \defeq U^* F U\otimes 1_{\Hils[L]}\)
is real, odd, and self-adjoint.
Let \(T\in \Cst_\Roe(\Z^d) \subseteq \Bound(\ell^2(\Z^d,\Hils[L]))\)
and \(S\in\Cliff_{0,d}\).
We must show that \((1-\tilde{F}^2) \cdot (T\otimes S)\)
and \([\tilde{F}^2, T\otimes S]\)
are compact operators.
The operator \(1 - \tilde{F}^2\)
acts by pointwise multiplication with \((1+\norm{n}^2)^{-1}\).
Since~\(T\)
is locally compact and~\(\Lambda^* \C^d\)
has finite dimension, the operator
\((1-\tilde{F}^2) \cdot (T\otimes S)\)
is compact.  Describe~\(T\)
as a block matrix \((T_{x,y})_{x,y\in\Z^d}\)
with \(T_{x,y}\in\Bound(\Hils[L])\).
The operator~\(\tilde{F}\)
anti-commutes with \(1\otimes S\).
So the graded commutator
\([A+A^*,T\otimes S] = [A+A^*,T\otimes 1]\cdot (1\otimes S)\)
corresponds to the block matrix with \((x,y)\)-entry
\[
T_{x,y} \otimes (\lambda_{x-y} + \lambda_{x-y}^*) S
\in \Bound(\Hils[L]\otimes \Lambda^* \C^d).
\]
Assume that~\(T\)
is \(R\)\nb-controlled,
that is, \(T_{x,y}=0\)
if \(\norm{x-y}>R\),
and that \(\sup_x \sum_y \norm{T_{x,y}}\)
and \(\sup_y \sum_x \norm{T_{x,y}}\)
are bounded; block matrices with these two properties give bounded
operators, and these are dense in the Roe \(\Cst\)\nb-algebra.
For such~\(T\),
the commutator \([A+A^*,T\otimes S]\)
satisfies analogous bounds because
\(\norm{\lambda_{x-y}+\lambda_{x-y}^*} \le 2 \norm{x-y} \le 2 R\)
whenever \(T_{x,y}\neq0\).
So the set of \(T\in \Cst_\Roe(X)\)
for which \([A+A^*,T \otimes S]\)
is bounded is dense in \(\Cst_\Roe(X)\).
Thus~\(A+A^*\)
defines a spectral triple over
\(\Cst_\Roe(X) \otimes \Cliff_{0,d}\).
As a consequence, \([\tilde{F}, T\otimes 1]\)
is compact for all \(T\in\Cst_\Roe(X) \otimes \Cliff_{0,d}\).
This finishes the proof of Theorem~\ref{the:fundamental_class}.

\subsection{Another topological artefact of the tight binding
  approximation}
\label{sec:artefact_tight_binding}

We already argued in the introduction that the tight binding
approximation may produce topological artefacts.  Namely, it suggests
to use the uniform Roe \(\Cst\)\nb-algebra
instead of the Roe \(\Cst\)\nb-algebra,
whose \(\K\)\nb-theory
is much larger.  We briefly mention another artefact caused by the
tight binding approximation.

We work in Bloch--Floquet theory for greater clarity.  The
Fermi projection of a Hamiltonian is described by a vector bundle
\(V\prto \T^d\)
over the \(d\)\nb-torus,
maybe with extra symmetries.  Here~\(d\)
is the dimension of the material, which is \(2\)
or~\(3\)
in the most relevant cases.  In \(\K\)\nb-theory,
two vector bundles \(\xi_1,\xi_2\)
are identified if they are \emph{stably isomorphic}, that is, there is
a trivial vector bundle~\(\vartheta\)
with \(\xi_1\oplus\vartheta \cong \xi_2\oplus\vartheta\).
Several authors put in extra work to refine the classification of
vector bundles (with symmetries) provided by \(\K\)\nb-theory
to a classification up to isomorphism, see
\cites{De_Nittis-Gomi:Real_Bloch, De_Nittis-Gomi:Quaternionic_Bloch,
  Kennedy:Thesis, Kennedy-Zirnbauer:Bott_gapped}.  Here we argue that
such a refinement of the classification is of little physical
significance.  The tight binding approximation leaves out energy bands
that are sufficiently far below the Fermi level.  Their inclusion only
adds a trivial vector bundle -- but this is the difference
between stable isomorphism and isomorphism.

\begin{theorem}[\cite{Husemoller:Fibre_bundles}*{Chapter 8, Theorem 1.5}]
  \label{the:K_stable_range}
  Let~\(X\)
  be an \(n\)\nb-dimensional
  CW-complex and let \(\xi_1\)
  and~\(\xi_2\)
  be two \(k\)\nb-dimensional
  vector bundles.  Let \(c=1,2,4\)
  depending on whether the vector bundles are real, complex or
  quaternionic.  Assume \(k \ge \lceil (n+2)/c \rceil - 1\).
  If \(\xi_1\)
  and~\(\xi_2\) are stably isomorphic, then they are isomorphic.
\end{theorem}

So for a \(3\)\nb-dimensional
space~\(X\),
the isomorphism and stable isomorphism classification agree for real
vector bundles of dimension at least~\(4\),
for complex vector bundles of dimension at least~\(2\),
and for all quaternionic vector bundles.
For instance, consider the material Bi$_2$Se$_3$ studied in
\cites{Zhang-Liu-Qi-Adi-Fang-Zhang:Insulators_BiSe,
  Liu-Qi-Zhang-Dai-Fang-Zhang:Model_Hamiltonian}.  The model
Hamiltonian in \cites{Zhang-Liu-Qi-Adi-Fang-Zhang:Insulators_BiSe,
  Liu-Qi-Zhang-Dai-Fang-Zhang:Model_Hamiltonian} focuses on four
bands, of which half are below and half above the Fermi energy.  But
the dimension of the physically relevant vector bundle is
\(2\cdot 83 + 3\cdot 34 = 268\),
the number of electrons per unit cell of the crystal; each atom of
Bismuth has 83~electrons and each atom of Se has 34~electrons.


The theorem cited above does not take into account a real involution
on the space~\(X\).
The proof of Theorem~\ref{the:K_stable_range} is elementary enough,
however, to extend to ``real'' vector bundles over ``real'' manifolds.
To see this, one first describes a ``real'' manifold as a
\(\Z/2\)\nb-CW-complex.
The main step in the proof of Theorem~\ref{the:K_stable_range} is to
build nowhere vanishing sections of vector bundles, assuming that the
fibre dimension is large enough.  This allows to split off a trivial
rank-\(1\)
vector bundle as a direct summand.  Similarly, if two vector bundles
with nowhere vanishing sections are homotopic, then there is a nowhere
vanishing section for the homotopy if the dimension of the fibres is
large enough.  The only change in the ``real'' case is that we need
a \(\Z/2\)\nb-equivariant
nowhere vanishing section of a ``real'' vector bundle to split off
trivial summands.  Such sections are built by induction over the
cells of the \(\Z/2\)\nb-CW-complex.
The \(\Z/2\)\nb-action
on the interior of such a cell is either free or trivial.  In the
first case, a \(\Z/2\)\nb-equivariant
section is simply a section on one half of the cell.  In the second
case, the cell is contained in the fixed-point submanifold, and we
need a nowhere vanishing section of a real vector bundle in the usual
sense.  So the argument in~\cite{Husemoller:Fibre_bundles} allows to
build nowhere vanishing real sections of ``real'' vector bundles
under the same assumptions on the dimension as for real vector bundles.

\end{document}